\documentclass[12pt,english]{article}
\usepackage{color}
\usepackage[T1]{fontenc}
\usepackage[latin9]{inputenc}
\usepackage[a4paper]{geometry}
\usepackage{textcomp}
\usepackage{amsthm}
\usepackage{amsmath}
\usepackage{amssymb}
\usepackage{esint}
\usepackage{bbm}
\usepackage{hyperref}
\usepackage{babel}

\makeatletter

\DeclareFontEncoding{LGR}{}{}
\DeclareTextSymbol{\~}{LGR}{126}

\theoremstyle{plain}
\newtheorem{thm}{Theorem}[section]
  \theoremstyle{definition}
  
  \theoremstyle{remark}
  \newtheorem{rem}[thm]{Remark}
  \theoremstyle{plain}
  \newtheorem{assumption}[thm]{Assumption}
  \theoremstyle{plain}
  \newtheorem{prop}[thm]{Proposition}
  \theoremstyle{definition}
  \newtheorem{defn}[thm]{Definition}
  \theoremstyle{plain}
  \newtheorem{lem}[thm]{Lemma}

\DeclareMathOperator*{\G}{d\Gamma}
\DeclareMathOperator*{\GG}{d\Gamma^{(2)}}

\makeatother

\begin{document}
\global\long\def\Tr{{\mathrm{Tr}}}
\global\long\def\Span{{\mathrm{span}}}
\global\long\def\fh{\mathfrak{h}}
\global\long\def\fhn{\mathfrak{H}_{f}^{(N)}}
\global\long\def\fhk#1{\mathfrak{H}_{f}^{(#1)}}
\global\long\def\om{\omega}
\global\long\def\Ex{\mathfrak{X}}
\global\long\def\dmn{\mathfrak{S}_{N}}
\global\long\def\dmna{\mathfrak{S}_{N_A}}
\global\long\def\dmnb{\mathfrak{S}_{N_B}}
\global\long\def\ghf{p}
\global\long\def\sobolev{{\mathrm{H}^1(\mathbb{R}^3)}}
  \global\long\def\g{\gamma}
\global\long\def\dx{d^{3}x}
\global\long\def\dy{d^{3}y}
\global\long\def\dz{d^{3}z}
 \global\long\def\Ran{\mbox{Ran}}
 \global\long\def\Id{\mathbbm{1}}
 \global\long\def\loss{}
 
 \newcommand{\ghfnot}{p} 
 \newcommand{\ghfnotp}{q} 
 \newcommand{\ghft}{p_t} 
 \newcommand{\ghftp}{q_t} 
 \newcommand{\Psio}{\Psi^{1}}
 \newcommand{\Psimo}{\Psi^{-1}} 
 \newcommand{\Psit}{\Psi^{2}}
 \newcommand{\Psimt}{\Psi^{-2}}
 \newcommand{\rhoj}{\rho^{[j]}}
 \newcommand{\rhomj}{\rho^{[-j]}}
 \newcommand{\rhoone}{\rho^{[1]}}
 \newcommand{\rhomone}{\rho^{[-1]}}
 \newcommand{\rhotwo}{\rho^{[2]}}
 \newcommand{\rhomtwo}{\rho^{[-2]}}
 \newcommand{\gammaj}{\gamma^{[j]}}
 \newcommand{\gammamj}{\gamma^{[-j]}}
 \newcommand{\gammaone}{\gamma^{[1]}}
 \newcommand{\gammamone}{\gamma^{[-1]}} 
 \newcommand{\gammatwo}{\gamma^{[2]}}
 \newcommand{\gammamtwo}{\gamma^{[-2]}}
 \newcommand{\Ij}{I^{[j]}}
 \newcommand{\Ione}{I^{[1]}}
 \newcommand{\Ionep}{I^{[1]\prime}}
 \newcommand{\Itwo}{I^{[-2]}}
 \newcommand{\Itwop}{I^{[-2]\prime}}
 \newcommand{\fj}{f^{[j]}}
 \newcommand{\fone}{f^{[1]}}
 \newcommand{\fonep}{f^{[-1]\prime}}
 \newcommand{\ftwo}{f^{[-2]}}
 \newcommand{\ftwop}{f^{[-2]\prime}}
 \newcommand{\KTF}{K_{TF}}

\title{Kinetic Energy Estimates for the Accuracy of the Time-Dependent Hartree-Fock
Approximation with Coulomb Interaction}

\author{
Volker Bach\footnote{Technische Universit\"at Braunschweig, Institut f\"ur Analysis und Algebra, Pockelsstra\ss e 14, {38106} Braunschweig, Germany. E-mail: {\tt v.bach@tu-bs.de}}, 
S\'ebastien Breteaux\footnote{BCAM - Basque Center for Applied Mathematics, Alameda de Mazarredo 14, {48009} Bilbao, Spain. E-mail: {\tt sbreteaux@bcamath.org}}, 
S\"oren Petrat\footnote{Institute of Science and Technology Austria (IST Austria), Am Campus 1, 3400 Klosterneuburg, Austria. E-mail: {\tt soeren.petrat@ist.ac.at}},
Peter Pickl\footnote{Ludwig-Maximilians-Universit\"at, Mathematisches Institut, Theresienstr.\ 39, {80333} M\"unchen, Germany. E-mail: {\tt pickl@math.lmu.de}},
Tim Tzaneteas\footnote{Universit\"at T\"ubingen, Fachbereich Mathematik, Auf der Morgenstelle 10, {72076} T\"ubingen, Germany. E-mail: {\tt timmy.tzaneteas@uni-tuebingen.de}}
}

\newpage
\maketitle

\begin{abstract}
We study the time evolution of a system of $N$ spinless fermions in
$\mathbb{R}^3$ which interact through a pair potential, e.g., the
Coulomb potential. We compare the dynamics given by the solution to
Schr\"odinger's equation with the time-dependent Hartree-Fock
approximation, and we give an estimate for the accuracy of this
approximation in terms of the kinetic energy of the system. This
leads, in turn, to bounds in terms of the initial total energy of the
system.
\end{abstract}

\textbf{MSC class:} 35Q40, 35Q55, 81Q05, 82C10


\section{Introduction}\label{sec:Introduction}
\paragraph{The Model.} 
In quantum mechanics, the state of a system of $N$ identical particles
is described by a wave function $\Psi_t$ which evolves in time $t \in
\mathbb{R}$ according to Schr\"odinger's equation,
\begin{equation}
\begin{cases}
i \partial_{t} \Psi_t & = H \Psi_t \, ,\\
\Psi_{t=0} & =\Psi_{0}\, .
\end{cases}\label{eq:Schrodinger-N-body}
\end{equation}
Given the (Bose-Einstein or Fermi-Dirac) particle statistics and the
one-particle Hilbert space $\fh$, the wave function $\Psi_t$ is a
normalized vector in $\mathfrak{H}_b^{(N)} := \mathcal{S}^{(N)}[
  \fh^{\otimes N}]$, for a system of $N$ bosons, or in $\fhn :=
\mathcal{A}^{(N)}[ \fh^{\otimes N}]$, for a system of $N$
fermions. Here $\mathcal{S}^{(N)}$ and $\mathcal{A}^{(N)}$ are the
orthogonal projections onto the totally symmetric and the totally
antisymmetric subspace, respectively, of the $N$-fold tensor product
$\fh^{\otimes N}$ of the one-particle Hilbert space $\fh$. The
dynamics~\eqref{eq:Schrodinger-N-body} is
 generated by the Hamilton
operator $H$ which is self-adjointly realized on a suitable dense
domain in $\mathfrak{H}_b^{(N)}$ or $\fhn$, respectively.

In the present article we study a system of $N$ spinless fermions in
$\mathbb{R}^3$, so $\Psi_t \in \fhn$, and $\fh = L^2[\mathbb{R}^3]$ is
the space of square-integrable functions on $\mathbb{R}^3$. The
Hamiltonian is given by
\begin{equation}
H \ = \ 
\nu + \sum_{j=1}^{N} h_{j}^{(1)} + \lambda\sum_{1\leq j<k\leq N}v(x_{j}-x_{k})\, ,
\label{eq:Hamiltonian}
\end{equation}
where
\begin{itemize}
\item the number $\nu \in \mathbb{R}$ is a constant contribution to
  the total energy. For example, if we describe a molecule in
  the Born-Oppenheimer approximation, then $\nu$ would account for the
  nuclear-nuclear repulsion,

\item  the coupling constant $\lambda >0$ is a small parameter and
  possibly depends on the particle number $N\geq1$ (while our interest ultimately lies in the description of systems with $N \gg 1$, the estimates in this article hold for any $N\geq1$),

\item  the self-adjoint operator $h^{(1)}$ on $\fh$  is of the form
  $-a \Delta + w(x)$, where $a>0$ and the external potential $w$ is an
  infinitesimal perturbation of the Laplacian,

\item 
and $v(x) := \pm |x|^{-1}$ is the Coulomb potential, for $x \in
  \mathbb{R}^3\setminus\{0\}$; $v(x) = +|x|^{-1}$ is the repulsive case, $v(x) = -|x|^{-1}$ the attractive case.
\end{itemize}
The Hamiltonian specified in \eqref{eq:Hamiltonian} describes several
situations of interest, e.g.:

\begin{itemize}
\item
\noindent\textbf{Atom.}
For an atom in the ($0^{th}$) Born-Oppenheimer approximation with a
nucleus of charge $Z$ at the origin, we have repulsive interaction and
\begin{equation}
\nu=0\,,\quad 
h^{(1)}=-\frac{\Delta}{2}-\alpha\frac{Z}{|x|}\,,\quad
\lambda=\alpha\,,
\label{eq:Def-Neutral-Atom}
\end{equation}
where $\alpha>0$ is the fine structure constant whose physical value
is $\alpha \simeq 1/137$. Note that our system of units is chosen such
that the reduced Planck constant $\hbar$, the electron mass $m$ and
the speed of light $c$ are equal to one, and the charge of the
electron is $-e=-\sqrt{\alpha}$. For more details about this choice of
units see~\cite[p.~21]{MR2583992}. 
\item
\noindent\textbf{Molecule.}
More generally, we can consider a molecule with $M \in
\mathbb{N}$ nuclei of charges $Z_{1}, \dots, Z_{M} >0$ at fixed,
distinct positions $R_{1}, \dots, R_{M} \in \mathbb{R}^3$ in the
Born-Oppenheimer approximation. In this case we have
\begin{align}
\nu  = \sum_{1\leq m<l\leq M} \frac{\alpha Z_{m}Z_{l}}{|R_{m}-R_{l}|}
\,,\quad 
h^{(1)} = -\frac{\Delta}{2} - \sum_{m=1}^{M}\frac{\alpha Z_{m}}{|x-R_{m}|}
\,,\quad
\lambda = \alpha \, .
\label{eq:Def-Molecule}
\end{align}
\item
\noindent\textbf{Particles in a Trap.}
For electrons in an external confining potential (realized, e.g., by a laser trap), we have repulsive interaction and
\begin{align}
\nu  = 0
\,,\quad 
h^{(1)} = -\frac{\Delta}{2} + w(x)
\,,\quad
\lambda = \alpha \,.
\label{eq:Def-Trap}
\end{align}
\item
\noindent\textbf{Fermion Star.} 
The Hamiltonian also describes systems of gravitating fermions, e.g., neutrons. In this case the interaction is attractive and
\begin{align}
\nu  = 0
\,,\quad 
h^{(1)} = -\frac{\Delta}{2}
\,,\quad
\lambda = G \,,
\label{eq:Def-Star}
\end{align}
where $G$ is Newton's gravitational constant (and recall that we set the mass $m=0$). A better description of a fermion star is achieved by replacing the non-relativistic Laplacian by the semi-relativistic operator $\sqrt{-\Delta+\Id}$.
\end{itemize}

For these situations the Hartree-Fock description that we are aiming at in this article and that we describe below can only be expected to hold for very short times (short relative to the large particle number $N$). For times of order $1$, we have to choose the coupling constant small in $N$ to see Hartree-Fock behavior (``mean-field scaling''). There are several possibilities to do that:

\begin{itemize}
\item
\noindent\textbf{Mean-Field Scaling for Large Volume.} 
Let us first note that for systems with large volume proportional to $N$, the kinetic energy is naturally also of order $N$. For such a system, the choice
\begin{equation}\label{eq:Def-Mean-Field-Scaling}
\nu = 0 ,\quad
h^{(1)} = -\frac{\Delta}{2} + w(x)
\, , \quad
\lambda = \frac{1}{N^{2/3}}
\end{equation}
leads to an interaction energy which is of the same order in $N$ as the kinetic energy (see \cite{Petrat:2014fk} for a more detailed discussion).
\item
\noindent\textbf{Mean-Field Scaling for Fixed Volume.} 
For systems with volume independent of $N$, the mean-field limit is naturally coupled to a semi-classical limit. Note that here the kinetic energy is of order $N^{5/3}$. Then the choice
\begin{equation}
\nu = 0, \quad  
h^{(1)} = -\frac{\Delta}{2N^{1/3}} + w(x)
\,,\quad
\lambda = \frac{1}{N^{2/3}}
\label{eq:Def-Semi-Classical-Mean-Field}
\end{equation}
leads to an interaction energy of the same order as the kinetic energy and nontrivial mean-field behavior (see in particular \cite{MR2092307,2013arXiv1305.2768B} for more details).
\item
\noindent$\boldsymbol{\lambda=N^{-1}}$ \textbf{Scaling.}
Very often, the term ``mean-field scaling'' is identified with the choice $\lambda=N^{-1}$. However, comparing with \eqref{eq:Def-Mean-Field-Scaling} and \eqref{eq:Def-Semi-Classical-Mean-Field}, in the two situations considered above, we see that this scaling leads to a subleading interaction.
\end{itemize}

\paragraph{Theory of the Time-Dependent Hartree-Fock Equation.}

Although (\ref{eq:Schrodinger-N-body}) admits the explicit solution
$\Psi_t = e^{-itH}\Psi_{0}$, this explicit form is not useful in
practice (from the point of view of numerics, for example) because of
the large number $N \gg 1$ of variables, and it therefore becomes
necessary to consider approximations to this equation. One such
approximation consists of restricting the wave function $\Psi_t$ to a
special class of wave functions. For fermion systems, the Hartree-Fock
approximation is a natural choice: it restricts $\Psi_t$ to the class
of Slater determinants, i.e., to those $\Phi \in \fhn$ which assume a
determinantal form,
\begin{equation}
\Phi(x_{1},\dots,x_{N})
=
\frac{1}{\sqrt{N!}}\det\left(\begin{array}{ccc}
    \varphi_{1}(x_{1}) & \cdots & \varphi_{1}(x_{N})\\ \vdots & \ddots &
    \vdots\\ \varphi_{N}(x_{1}) & \cdots &
    \varphi_{N}(x_{N})\end{array}\right)\ ,
\label{eq:Slater-determinant}
\end{equation}
where the orbitals $\varphi_{1}, \ldots, \varphi_{N} \in \fh$ are
orthonormal. We express (\ref{eq:Slater-determinant}) more concisely
as $\Phi = \varphi_{1} \wedge \cdots \wedge \varphi_{N}$. In
time-independent Hartree-Fock theory, one is interested in determining
the minimal energy expectation when varying solely over Slater determinants
 \cite{MR1175492, MR1297873, BachFroehlichJonsson2009,
  MR2781970, Bach2013}, i.e., one is interested in finding
\[
\inf\big\{ \langle\Phi,H\Phi\rangle \, \big| \, 
\Phi = \varphi_{1}\wedge\cdots\wedge\varphi_{N} \,, \quad
\langle \varphi_{i}, \varphi_{j} \rangle = \delta_{ij} \big\} \,.
\]

One can also study the evolution governed by
(\ref{eq:Schrodinger-N-body}) using Slater determinants, which gives
rise to time-dependent Hartree-Fock theory. Here the basic intuition
is that, for a system containing a large number of particles, the
solution will stay close to a Slater determinant (at least for short
times), provided the initial state is close to a Slater
determinant. Turning this intuition into mathematics requires the
specification of the equation of motion of the approximating Slater
determinant, as well as a mathematically rigorous notion of being
``close''. For the derivation of the former, one assumes that
the solution to (\ref{eq:Schrodinger-N-body}) is of the form $\Phi_t =
\varphi_{t,1} \wedge \cdots \wedge \varphi_{t,N}$, as in
\eqref{eq:Slater-determinant}. It is then easy to verify that the
orbitals $\varphi_{t,1},\dots,\varphi_{t,N}$ necessarily satisfy the
time-dependent Hartree-Fock (TDHF) equation, that is the system of $N$
non-linear equations given by
\begin{equation}
i\frac{d\varphi_{t,j}}{dt} \ = \ 
h^{(1)}\varphi_{t,j} + \lambda \sum_{k=1}^{N} 
\Big( [v*|\varphi_{t,k}|^{2}]\varphi_{t,j} - 
[v*(\varphi_{t,j} \bar{\varphi}_{t,k})]\varphi_{t,k} \Big) \, ,
\label{eq:hartree-fock_tradi}
\end{equation}
for $j = 1, \dots, N$ ($\bar{\varphi}$ is the complex conjugate of $\varphi$).

The TDHF equation \eqref{eq:hartree-fock_tradi} can be rewritten
in terms of the one-particle density 
matrix~$\ghf_t=\sum_{j=1}^{N}|\varphi_{t,j}\rangle\langle\varphi_{t,j}|$
with $\varphi_{t,j} \in\fh$ and 
$\langle\varphi_{t,j},\varphi_{t,k}\rangle=\delta_{jk}$
as
\begin{equation}
(\mbox{TDHF})\qquad 
i\partial_{t}\ghf_t \ = \ 
[h^{(1)},\ghf_t] + 
\lambda \Tr_{2}[v^{(2)},(\ghf_t\otimes\ghf_t)(\Id-\Ex)]\,.
\label{eq:hartree-fock_density_matrix}
\end{equation}
Here $\Ex$ is the linear operator on $\fh\otimes\fh$ such that
$\Ex(\varphi\otimes\psi)=\psi\otimes\varphi$ and $\Tr_{2}$ is the
partial trace (see~\eqref{eq:def-partial-trace}). Sometimes, we write $\ghf_t^{(2)} =
(\ghf_t\otimes\ghf_t)(\Id-\Ex)$. In the sequel, when
speaking of the TDHF equation, we refer to
\eqref{eq:hartree-fock_density_matrix}. The term involving $\Id$ is called the \emph{direct term}, the term involving $\Ex$ the \emph{exchange term}. 

Note that the TDHF equation \eqref{eq:hartree-fock_density_matrix} can
be written as $i\partial_{t}\ghf_t =
[h_{HF}^{(1)}(\ghf_t) , \ghf_t]$, where the effective
HF-Hamiltonian $h_{HF}^{(1)}(\gamma)$ is given by
\begin{equation}
h_{HF}^{(1)}(\gamma) \ := \ 
h^{(1)} + \lambda \Tr_{2}[v^{(2)}(\Id_{\fh\otimes\fh}-\Ex)(\Id_\fh \otimes \gamma)]\,.
\label{eq:hartree-fock_HF-Hamiltonian}
\end{equation}
Implicitly assuming the existence and regularity of $\ghf_t$,
the HF-Hamiltonian $h_{HF}^{(1)}(\ghf_t)$ is self-adjoint with
the same domain as $h^{(1)}$, and hence the solution to $\partial_{t}
U_{HF,t} = -i h_{HF}^{(1)}(\ghf_t) U_{HF,t}$, with $U_{HF,0} =
\Id$, is unitary. This has the important consequence that
\eqref{eq:hartree-fock_density_matrix} preserves the property of the
one-particle density matrix $\ghf_t$ of being a rank-$N$
orthonormal projection. In other words, if $\Phi_t \in \fhn$ evolves
according to the TDHF equation and $\Phi_0 = \varphi_{1} \wedge \cdots
\wedge \varphi_{N}$ is a Slater determinant, then so is $\Phi_t$, for all
$t \in \mathbb{R}$.

The TDHF equation for density matrices as in \eqref{eq:hartree-fock_density_matrix} 
has been studied in \cite{MR0424069} for a bounded two-body interaction. 
Then the mild solutions of the TDHF equation in the form \eqref{eq:hartree-fock_tradi} have been handled for a Coulomb two-body potential in \cite{MR0413843} for initial data in the Sobolev space $\mathrm{H}^1$. 
This result has been extended to the TDHF equation in the form 
\eqref{eq:hartree-fock_density_matrix} in \cite{MR0456066,MR0411439}. Note that \cite{MR0456066} also handles the case of a more general class of two-body potentials and the existence of a classical solution for initial data in a space similar to the Sobolev $\mathrm{H}^2$ space for density matrices. In \cite{MR1175475} the existence of mild solutions  of the TDHF in the form \eqref{eq:hartree-fock_tradi} was proved for a Coulomb two-body potential with an (infinite sequence of) initial data in $L^2$. For the convenience of the reader we state the precise results we use about the theory of the TDHF equation in Appendix~\ref{sec:Theory_TDHF}. In \cite{zbMATH00887075} the existence and uniqueness  of strong solutions to the von Neumann-Poisson equation, another nonlinear self-consistent time-evolution equation on density matrices, are proved with the use of a generalization of the Lieb-Thirring inequality. Another direction in which to generalize the Hartree equations is to consider, instead of an exchange term, a dissipative term in the Hartree equation; the existence and uniqueness of a solution for such an equation has been proved in \cite{MR2096738}.

\paragraph{One-particle Density Matrix.}
The notion of proximity of two states we use in this article is defined
by expectation values of $k$-particle observables, where $1 \leq k \ll
N$. More specifically, if $\Psi_{t} \in \fhn$ is the (normalized)
solution to (\ref{eq:Schrodinger-N-body}) and $\Phi_{HF,t} =
\varphi_{t,1} \wedge \cdots \wedge \varphi_{t,N}$, where $\varphi_{t,1}, \dots,
\varphi_{t,N}$ are the solutions to (\ref{eq:hartree-fock_tradi}), then,
for any $k$-particle operator $A^{(k)}$ (i.e., for any bounded
operator $A^{(k)}$ on $\fh^{\wedge k} := \mathcal{A}[\fh^{\otimes k}]$), 
we wish to control the quantity
\begin{align*}
\delta_{t}^{(k)}\big(A^{(k)}\big) 
\ := \
\frac{1}{\|A^{(k)}\|_{\infty}}
\big| \langle \Psi_{t} , \; (A^{(k)}\otimes \Id_{N-k}) \Psi_{t} \rangle
- \langle \Phi_{HF,t} , \; (A^{(k)}\otimes\Id_{N-k}) \Phi_{HF,t} \rangle
\big| \,.
\end{align*}
Here $\Id_{N-k}$ denotes the identity operator on $\fh^{\otimes
  (N-k)}$ and $\|\cdot\|_{\infty}$ denotes the operator norm on
$\mathcal{B}[\fh^{\wedge k}]$.

It is more convenient to reformulate this approach in terms of reduced
density matrices. We recall that, given $\Psi\in\fhn$,
the corresponding reduced $k$-particle density matrix is the
trace-class operator $\gamma_{\Psi}^{(k)}$ on $\fhk {k}$ whose kernel is
given by
\begin{multline}\label{eq:p-particle_density_matrix}
\gamma_{\Psi}^{(k)}(x_{1},\dots,x_{k};y_{1},\dots y_{k})
\\ =
\frac{N!}{(N-k)!} \int\overline{\Psi(x_{1},\dots x_{k},x_{k+1},\dots x_{N})}
\,\Psi(y_{1},\dots y_{k},x_{k+1},\dots x_{N}) \; \dx_{k+1} \cdots \dx_{N}\,.
\end{multline}
Note that we normalize the reduced density matrices so that
$\Tr\gamma_{\Psi}^{(k)}=\frac{N!}{(N-k)!}$.  We may then rewrite
$\delta_{t}^{(k)}(A^{(k)})$ as
\[
\delta_{t}^{(k)}\big(A^{(k)}\big)
=
\frac{1}{\|A^{(k)}\|_{\mathcal{B}(\fhk{{k}})}}
\Big| \Tr\big[ (\gamma_{\Psi_{t}}^{(k)} - \gamma_{\Phi_{HF,t}}^{(k)})A^{(k)} 
\big] \Big|
\]
and observe that 
\[
\sup_{A^{(k)}\in\mathcal{B}(\fhk{{k}})} \delta_{t}^{(k)}\big(A^{(k)}\big)
=
\big\| \gamma_{\Psi_{t}}^{(k)} - 
\gamma_{\Phi_{{HF,t} }}^{(k)} \big\|_{\mathcal{L}^{1}}\,,
\]
where $\|\cdot\|_{\mathcal{L}^{1}}$ denotes the trace norm. We are thus interested
in bounds on $\| \gamma_{\Psi_{t}}^{(k)} -
\gamma_{\Phi_{HF,t}}^{(k)}\|_{\mathcal{L}^{1}}$. In the present article we restrict
ourselves to the case $k=1$. 

\paragraph{Derivation of the TDHF Equation.}

The derivation of the TDHF equation
  may be seen as part of the quest for a derivation of macroscopic, or
  mesoscopic, dynamics from the microscopic classical or quantum-mechanical dynamics of many-particle systems as an effective theory. Let us first discuss some generally interesting examples and then come to the case of the TDHF equation for fermions.

In the case of the dynamics of $N$ identical quantum-mechanical particles, the time-dependent Hartree equation, that is the TDHF equation \eqref{eq:hartree-fock_tradi} without the exchange term, was first derived rigorously in~\cite{MR578142} for a system of $N$ distinguishable particles in the mean-field limit. For systems of indistinguishable particles, the case of bosons has received considerable attention compared to the case of fermions, and several methods have been developed. The so-called Hepp method has been developed in~\cite{MR0332046,MR530915,MR539736} in order to study the classical limit of quantum mechanics. It inspired, among others, \cite{zbMATH01645814}, where the convergence to the Hartree equation is proved, \cite{MR2530155}, where the rate of convergence toward mean-field dynamics is studied, and \cite{MR2465733,MR2513969}, where the propagation of Wigner measures in the mean-field limit is studied, with special attention to the relationships with microlocal and semiclassical analysis. In this direction, with a stochastic microscopic model, the linear Boltzmann equation was obtained as a weak-coupling limit in~\cite{Breteaux20141037} yielding an example for a derivation of an equation with non-local terms using methods of pseudodifferential calculus. The derivation of the linear Boltzmann equation in the earlier work~\cite{MR1744001}, along with the series of works following it, used  a different method based on series expansions in terms of graphs similar to Feynman diagrams. The result is valid on longer time-scales than in \cite{Breteaux20141037}, but with more restrictive initial data. Other limit dynamics have been obtained, a particularly interesting one is the weak-coupling limit for interacting fermions for which a (non-rigorous) derivation of the nonlinear Boltzmann equation has been given in~\cite{MR2083147}. Series expansion methods and the Bogoliubov-Born-Green-Kirkwood-Yvon (BBGKY) hierarchy have also proved fruitful in other works, e.g.,  \cite{MR578142,MR1869286,MR2209131,MR2331036,MR2276262,MR2802894,MR2680421,MR2525781}. In \cite{MR2680421,MR2525781} the Gross-Pitaevskii equation, which describes the dynamics of a Bose-Einstein condensate has been derived. Also for the Gross-Pitaevskii equation the formation of correlations has been studied in \cite{MR2511665}, providing information on the structure of solutions to the Gross-Pitaevskii equation. The techniques developed in \cite{MR2755061} to study the weakly nonlinear Schr\"odinger equation are used in \cite{MR2518987} to derive quantum kinetic equations; those techniques resemble the BBGKY hierarchy methods, but they do not impose the normal ordered product of operators when considering expectation with respect to the initial state. The bounds on the rate of convergence in the mean-field limit given in \cite{zbMATH01645814} have been sharpened in \cite{MR2518972} using a method inspired by Lieb-Robinson inequalities. Another method introduced in \cite{MR2313859} shows that the classical time evolution of observables commute with the Wick quantization up to an error term which vanishes in the mean-field limit, yielding an Egorov-type theorem. Recently a new method based on a Gr\o nwall lemma for a well-chosen quantity has been introduced \cite{MR2821235,MR2657816} in the bosonic case, which considerably simplifies the convergence proof for the Hartree equation.

In the fermionic case, the TDHF equation has been derived in \cite{MR1996777} in the 
$\lambda=N^{-1}$ scaling for initial data close to Slater determinants, and
with bounded two-body potentials. The same authors give bounds on the accuracy
of the TDHF approximation for uncorrelated initial states in
\cite{MR2054172}, still with a bounded two-body potential. For the same scaling, the TDHF equation has been derived in \cite{MR2841931} for the
Coulomb potential for sequences of initial states given by Slater determinants. 
The semi-classical mean-field scaling from \eqref{eq:Def-Semi-Classical-Mean-Field} has first been considered in \cite{narnhofer:1981} where it is shown that for suitably regular interactions the Schr\"odinger dynamics is close to the classical Vlasov dynamics. The results have been improved in \cite{spohn:1981}. 
In \cite{MR2092307}, in the semi-classical mean-field scaling, the closeness of the Schr\"odinger dynamics to the Hartree-Fock dynamics was discussed and bounds for
the Husimi function were given, assuming the
potential to be real-analytic and thus in particular bounded. Up to that point all the method used to derive the TDHF equation had always been based on BBGKY hierarchies. In
\cite{2013arXiv1305.2768B,2013arXiv1311.6270B} estimates of
$\|\gamma_{N,t}-\ghf_{N,t}\|_{\mathcal{L}^{1}}$ were given in terms
of the number $N$ of electrons and the time $t$, in the semi-classical mean-field scaling. Their method is based on the Gr\o nwall lemma, similarly to \cite{MR2821235} in the bosonic case. The second article deals with the semi-relativistic case. The authors pointed out that with a bounded potential, in this scaling, the exchange term in the
time-dependent Hartree-Fock equation does not play a role so that the time-dependent Hartree-Fock equation reduces to the time-dependent Hartree
equation.
In \cite{Petrat:2014fk}, the fermionic Hartree equation in the large volume case is considered by generalizing the method of \cite{MR2821235}. Interactions of the form $|x|^{-s}$ are considered, with the corresponding $\lambda=N^{-1+s/3}$. Under the condition that the Hartree-Fock kinetic energy per particle is bounded uniformly in time, a derivation of the TDHF equation is given for $0<s<3/5$, and for Coulomb interaction with either a mild singularity cutoff on a ball with radius $N^{-1/6+\varepsilon}$, for any $\varepsilon>0$, or for the full Coulomb interaction under certain Sobolev conditions on the solution to the TDHF equation which are not proven in this work. Explicit bounds in terms of $N$, the Hartree-Fock kinetic energy and $t$ are given. Furthermore, in \cite{Petrat:2014fk}, the main result of \cite{2013arXiv1305.2768B} is reproduced with a different method than in \cite{2013arXiv1305.2768B} and written down for weaker conditions on the closeness of the initial state to a Slater determinant.

\paragraph{Main Estimate of this Article (see Theorem~\ref{thm:estimate-trace-norm}).} 
Given a normalized initial state $\Psi_{0} \in \fhn$ and the one-particle density matrix
  $\ghf_0 \equiv \gamma_{\Phi_{HF,0}}$ associated with a Slater
  determinant $\Phi_{HF,0} = \varphi_{1,0} \wedge \cdots \wedge
  \varphi_{N,0}$, with $\langle \varphi_{i,0}, \varphi_{j,0} \rangle$
  being orthonormal orbitals in $\sobolev$, $\gamma_t$ the one-particle density matrix of the solution $\Psi_t$ 
   to \eqref{eq:Schrodinger-N-body} and $\ghf_t$ the solution to
  \eqref{eq:hartree-fock_density_matrix} obey the trace norm estimate 
\begin{equation}\label{eq:main_estimate}
\frac{1}{N} \|\gamma_{t}-\ghf_t\|_{\mathcal{L}^{1}} \leq \sqrt{8}\sqrt{N^{2/3} \frac{1}{N} \|\gamma_{0}-\ghf_0\|_{\mathcal{L}^{1}} \exp (C_{\lambda,N,K} t)+ N^{-1/3} \big( \exp (C_{\lambda,N,K} t)-1 \big) } \, ,
\end{equation}
with $C_{\lambda,N,K} = 30 \lambda \sqrt{K} N^{1/6}$, where $K$ is a bound on the kinetic energy of $\ghf_t$ which is
assumed to be uniform in time (see
\eqref{eq:K-uniformly-bounded}).

\paragraph{Discussion of the Results.}
Roughly speaking, the estimate \eqref{eq:main_estimate} implies that, starting from a state close to a Slater determinant for the $N$-body Schr\"odinger equation and from the corresponding  one-particle density matrix for the TDHF equation, the Hartree-Fock approximation is justified up to times of order~$(\lambda \sqrt{K} N^{1/6})^{-1}$, where $K$ is the kinetic energy (which, for repulsive systems, is bounded by the total energy of the system, uniformly in time) and $\lambda$ the coupling constant. Hence, our assumption on the initial state is given in terms of energy, and not in the form of ``increasing'' sequences of Slater determinants. This assumption seems more natural to the authors as it is closer to a thermodynamic assumption on the system. In our proof we obtain a rate of convergence of $N^{-1/6}$. For the initial data, in order to have convergence, we can allow states with $N^{-1/3}\|\gamma_{0}-\ghf_0\|_{\mathcal{L}^{1}} \to 0$ for $N\to\infty$. This means, e.g., that, for any $\varepsilon>0$, the initial state can have $N^{1/3-\varepsilon}$ particles outside the condensate, i.e., the Slater determinant structure.

The fact that the estimate~\eqref{eq:main_estimate} is relevant when $\lambda N^{1/6}K^{1/2}t$ is of order one, restricts its applicability to a regime where the kinetic energy dominates the direct and exchange terms. This implies that the evolution is the free evolution to leading order. Estimate~\eqref{eq:main_estimate} captures the subleading effect of the direct term on the dynamics and is thus relevant provided that $K\gg N^{4/3}$. We substantiate this by heuristic arguments in Appendix~\ref{sec:misc_appendix}. Let us stress that Estimate~\eqref{eq:main_estimate} requires no additional assumption on the initial states other than the Hartree-Fock kinetic energy to be finite. 
Furthermore, Estimate~\eqref{eq:main_estimate} applies to the repulsive or attractive Coulomb interaction, which is very relevant for many physical systems.

Compared to \cite{MR2841931}, where also the Coulomb potential was considered, our result holds for larger time scales.
In \cite{MR2841931}, the $\lambda=N^{-1}$-scaling was assumed and, by a rescaling in time and in space, the result also applies to a large neutral atom (i.e., with charge $N\gg 1$ and $\lambda=\alpha$). With the result of~\cite{MR2841931} the Hartree-Fock approximation is then justified up to times of order $N^{-2}$. Assuming we have a state with a  negative energy, the kinetic energy is controlled by a universal multiple of $N^{7/3}$ (see Sect.~\ref{sec:Main-Result} for more details), and our estimate allows us to justify the approximation up to much larger times, of order $N^{-4/3}$. (Note, however, that our estimate deteriorates if the energy of the state is higher.)

Compared to \cite{2013arXiv1305.2768B} where the semi-classical scaling \eqref{eq:Def-Semi-Classical-Mean-Field} is considered, our result allows us to control the approximation only up to times of order $N^{-1/3}$, whereas the estimates in~\cite{2013arXiv1305.2768B} allow one to control the approximation up to times of order $1$ (however, only for bounded two-body potentials). This comes from the fact that we do not assume any semi-classical structure on the initial data. Note that our strategy is similar to the one of~\cite{2013arXiv1305.2768B} since we do not use the BBGKY hierarchy but instead make use of a Gr\o nwall lemma. An important difference lies in the decomposition of the potential: in~\cite{2013arXiv1305.2768B} a Fourier decomposition is used whereas we use the Fefferman-de la Llave formula.

Let us compare our results to \cite{Petrat:2014fk} where the mean-field scaling for large volume \eqref{eq:Def-Mean-Field-Scaling} is considered. Note that there the Schr\"odinger dynamics is compared to the fermionic Hartree equation without exchange term. While in \cite{Petrat:2014fk} other interactions are also considered, for Coulomb interaction, essentially two results are proven. First, for regularized Coulomb interaction with singularity cut off on a ball with radius $N^{-1/6+\varepsilon}$ for any $\varepsilon>0$, convergence of the Schr\"odinger dynamics to the fermionic Hartree dynamics is shown in terms of a counting measure $\alpha_g$, with convergence rate depending on the cutoff. Note, that we use the same measure in our proof, see also Remark~\ref{rem:spectral_decomp}, but we formulate our main result only in terms of the trace norm difference of reduced densities. The improvement of our result is that it holds for full Coulomb interaction without any regularization and, in general, with a better convergence rate. For the second result in \cite{Petrat:2014fk} a bound on $\Tr[(-\Delta)^{3+\varepsilon}\ghf_t]$ is assumed. Under that condition convergence for full Coulomb interaction in terms of $\alpha_g$ and the trace norm difference is shown, with rate $N^{-1/2}$ in the trace norm sense. This bound on $\Tr[(-\Delta)^{3+\varepsilon}\ghf_t]$ was, however, not proven to hold for $t>0$. Compared to that, our result holds for any initial condition with kinetic energy bounded by $C N$, without further assumptions, but only with a convergence rate of $N^{-1/6}$ in the trace norm sense.

\paragraph{Sketch of our Derivation of Estimates on the Accuracy 
of the TDHF Approximation.}

We derive an estimate on the trace norm of the difference $\gamma_t-\ghf_t$
between the one-particle density matrix $\gamma_{t} \equiv
\gamma_{\Psi_t}$ of the (full) solution $\Psi_t = e^{-itH} \Psi_0$ of
\eqref{eq:Schrodinger-N-body} and the one-particle density matrix
$\ghf_t$ solving the TDHF equation
\eqref{eq:hartree-fock_density_matrix}. 
Our work is inspired by \cite{MR2821235}, where one of us developed a new method for bosonic systems which was generalized to fermion systems in \cite{Petrat:2014fk} by two of us. The method uses a Gr\"onwall estimate for a well-chosen quantity called the \textit{number of bad particles} in \cite{MR2821235}. We refer to the quantity we chose to control as the \textit{degree of evaporation} $S_g$. The subscript $g$ refers to a freedom in the choice of a weight function $g$ which allows us to fine-tune the distance of $\rho_t$ (the density matrix of $\Psi_t$) to $\ghft$ in a suitable way. For the simplest choice $g(x)=x$, $S_g$ is called the \textit{degree of non-condensation} in \cite[Remark~(a) on p.~5]{MR1301362}, while in \cite{1751-8121-42-8-085201} it is called \textit{Verdampfungsgrad}, which translates to \textit{degree of evaporation}.
 
We show that the degree of evaporation $S_g$ is directly related to the trace norm $\|\gamma-\ghf\|_{\mathcal{L}^{1}}$. 
We then calculate the time derivative of $S_g$ and split it into three terms that we estimate separately.
To obtain the estimates
we make
use of correlation inequalities which  may be
  seen to be a dynamical version of the correlation
estimate presented in
\cite{MR1175492}. (See also \cite{MR1301362} for
      an
       alternative proof of  that
     correlation estimate which does not make
       use of second quantization.) 
While we estimate two of the terms in a way very similar to \cite{Petrat:2014fk}, our estimate for the remaining term (here called $\mathcal{A}$; in \cite{Petrat:2014fk} called $(I)$) is very different and allows us to treat the full Coulomb potential. This term is physically the most important, since its smallness is a consequence of cancellations between the Hartree-Fock and the many-body interaction. The bounds on this term are the key estimates of this work. They are obtained by using the Fefferman-de la Llave decomposition formula
     \cite{MR864658}.  We
       remark that, in view of the generalization of this
       decomposition derived in \cite{MR1930084,MR1701408}, our
     result applies to a more general class of two-body interaction
     potentials. The Lieb-Thirring
     inequality~\cite{PhysRevLett.35.687} and Hardy's inequality then provide an estimate in
     terms of kinetic energy. Finally, we note that in many physically relevant
     cases the estimate in terms of kinetic energy can be stated in
     terms of an estimate on the initial total energy of the system.

\paragraph{Outline of the Article.}

In Sect.~\ref{sec:Main-Result} we state our main result, along with applications to molecules or the mean-field limit. In Sect.~\ref{sec:Relative-Entropy}
we introduce the degree of evaporation $S_g$ and relate it to the difference
between the one-particle density matrix of the solution to our model
and the solution to the TDHF equation. We then calculate the time derivative of $S_g$ and provide bounds for the different contributions, thus proving our main theorem.
In Appendix~\ref{sec:Theory_TDHF} we recall some results about 
the theory of the TDHF equation.

\section{Main Result and Applications}\label{sec:Main-Result}

Our main result is an estimate of the trace norm $\| \cdot \|_{\mathcal{L}^{1}}$ of the difference between the
one-particle density matrix of the solution to the many-body
Schr\"odinger equation \eqref{eq:Schrodinger-N-body} and the solution to
the time-dependent Hartree-Fock
equation \eqref{eq:hartree-fock_density_matrix} in terms of the
kinetic energy of the system. As usual, we denote by $\sobolev$ the
Sobolev space of weakly differentiable functions with
square-integrable derivative.

We henceforth make use of the following notation:
\begin{itemize}
\item
Let $\Psi_{0} \in \fhn$ be a normalized initial state, and let $\gamma_{t} := \gamma_{\Psi_{t}}$ be the one-particle density matrix of the solution $\Psi_{t} = e^{-iHt} \Psi_{0}$ to the Schr\"odinger equation \eqref{eq:Schrodinger-N-body} with Hamiltonian $H$ from \eqref{eq:Hamiltonian} (i.e., with Coulomb interaction).
\item 
Let $\Phi_{HF,0} = \varphi_{1,0} \wedge
  \cdots \wedge \varphi_{N,0}$ be a Slater determinant, with $\varphi_{j,0}
  \in \sobolev$ and
  $\langle\varphi_{j,0},\varphi_{k,0}\rangle_{\fh}=\delta_{jk}$, for $1
  \leq j,k \leq N$. Let $\ghf_0 := \gamma_{\Phi_{HF,0}}$ be the
  one-particle density matrix of $\Phi_{HF,0}$ and 
  $\ghf_t$ be the solution to the time-dependent Hartree-Fock
  equation \eqref{eq:hartree-fock_density_matrix} with initial
  condition $\ghf_0$.
\end{itemize}
  \begin{thm}
\label{thm:estimate-trace-norm}
Assume that the kinetic energy of $\ghf_t$ is uniformly bounded in time,
\begin{equation}
K \ := \ 
\sup_{t\geq0} \Tr[-\Delta\ghf_t]
\ < \ \infty \,.
\label{eq:K-uniformly-bounded}
\end{equation}
Under the assumption of \eqref{eq:K-uniformly-bounded}
the estimate
\begin{equation}\label{eq:main_result}
\frac{1}{N} \|\gamma_{t}-\ghf_t\|_{\mathcal{L}^{1}} \leq \sqrt{8}\sqrt{N^{2/3} \frac{1}{N} \|\gamma_{0}-\ghf_0\|_{\mathcal{L}^{1}} \exp (C_{\lambda,N,K} t)+ N^{-1/3} \big( \exp (C_{\lambda,N,K} t)-1 \big) } 
\end{equation}
holds true with $C_{\lambda,N,K} = 30 \lambda \sqrt{K} N^{1/6}$.
\end{thm}
The proof of Theorem~\ref{thm:estimate-trace-norm} is postponed to Sect.~\ref{sec:Relative-Entropy}.
\begin{rem}
One of the ingredients of our proof is the Fefferman-de la Llave
decomposition of the Coulomb potential \cite{MR864658}
\begin{equation}
\frac{1}{|x|} \ = \ 
\int_{0}^{\infty} \frac{16}{\pi \, r^5} \,(1_{B(0,r/2)}*1_{B(0,r/2)})(x) \, dr\,,
\label{eq:Feffermann-de_la_Lave}
\end{equation}
an identity that holds for all $x \in \mathbb{R}^{3}\setminus\{0\}$,
where $1_{B(0,r/2)}$ is the characteristic function of the
ball of radius $r/2$ centered at the origin in $\mathbb{R}^{3}$. A
generalization of this decomposition to a class of two-body
interaction potentials $v$ of the form
\begin{equation}
v(x) \ = \ \int_{0}^{\infty}g_{v}(r)\,(1_{B(0,r/2)}*1_{B(0,r/2)})(x)\, dr\,,
\label{eq:Hainzl-Seiringer}
\end{equation}
with $x \in \mathbb{R}^{3}\setminus\{0\}$, was given in
\cite{MR1930084} under Assumption \ref{ass:Fefferman-delaLlave} below, and our proof largely generalizes to those
potentials $v$. More precisely, the assertion of
Theorem~\ref{thm:estimate-trace-norm} holds true and without any
change in the constants, if we replace the Coulomb potential by any
pair potential $v$ that satisfies
Assumption~\ref{ass:g(r)leq_r-5} below, which in particular implies $v(x) \leq |x|^{-1}$. Note that the assumption of semi-boundedness of $v$ is only used to ensure the global existence of a solution to the TDHF equation. One could drop it to study problems up to the time the solution to the TDHF blows up.
\end{rem}

\begin{assumption}
\label{ass:Fefferman-delaLlave} 
The function $v: \mathbb{R}^{3}\setminus\{0\} \to \mathbb{R}$ has the following properties:
\begin{itemize}
\item $v$ is a radial function, and there exists a function
  $\tilde{v} \in C^3[ (0,\infty); \mathbb{R}]$ such that $v(x) =
  \tilde{v}\big(|x|\big)$, for all $x \in \mathbb{R}^{3} \setminus
  \{0\}$,

\item $r^{m}\, \frac{d^m \tilde{v}}{dr^m}(r) \to 0$, as 
$r \to \infty$, for $m=0,1,2$,

\item $\lim_{R\to\infty} \int_{1}^{R} r^{3} \, g_{v}(r) \, dr$ exists,
  with $g_{v}(r) := \frac{2}{\pi} \frac{d}{dr} \big( \frac{1}{r}
  \frac{d^2 \tilde{v}}{dr^2}(r) \big)$.
\end{itemize}
\end{assumption}
\noindent
Note that $g_{|\cdot|^{-1}}(r) = \frac{16}{\pi} r^{-5}$ in case of the
Coulomb potential which is prototypical for the following further
assumption.
\begin{assumption}\label{ass:g(r)leq_r-5}
(With the same notation as in
Assumption~\ref{ass:Fefferman-delaLlave}.) The function $v:\mathbb{R}^{3}\to\mathbb{R}$  satisfies Assumption~\ref{ass:Fefferman-delaLlave}, $|g_{v}(r)|\leq
\frac{16}{\pi} r^{-5}$ and, for some $\mu\in \mathbb R$, $v(x)\geq \mu$ for all $x$.
\end{assumption}

\begin{rem}
Note that we actually prove a slightly stronger result in Theorem~\ref{thm:estimate-S} in terms of the degree of evaporation $S_g(t)$ (with properly chosen $g$), which is defined in Definition~\ref{def:S}. The way our result is formulated in Theorem~\ref{thm:estimate-S} can directly be compared to the results in \cite{Petrat:2014fk}.
\end{rem}
\begin{rem}
Note that the two summands in the square root on the right-hand side of \eqref{eq:main_result} come from different contributions which we call $\mathcal{A}_t$, $\mathcal{B}_t$ and $\mathcal{C}_t$ (and which are called $(I)$, $(II)$, $(III)$ in \cite{Petrat:2014fk}), see Proposition~\ref{pro:equation-dS_g/dt}. It is interesting to note that all three terms contribute to the first summand (which is proportional to $\|\gamma_{0}-\ghf_0\|_{\mathcal{L}^{1}}$) but only the $\mathcal{B}_t$ term contributes to the second summand.
\end{rem}
\begin{rem}\label{rem:proof_in_H1_is_enough}
Note that it is sufficient to prove Theorem~\ref{thm:estimate-trace-norm} with $\Psi_0$ in $\fhn \cap \sobolev^{\otimes N}$. A density argument then provides the result for a general $\Psi_0$ in $\fhn$.
\end{rem}

Let us discuss some cases when the assumption that the Hartree-Fock kinetic energy is uniformly bounded in time is satisfied. In Propositions \ref{cor:large-atoms} and \ref{cor:mean-field} we give
explicit bounds on the kinetic energy $K$ in terms of the energy
expectation value $\langle \Phi_{HF,0}, \, H \Phi_{HF,0} \rangle$ of the initial state $\Phi_{HF,0}$ and the ground state
energy for examples presented in Sect.~\ref{sec:Introduction}.
In the case of atoms or molecules this follows from known estimates, which we
now recall.

To formulate these, we denote the energy expectation value 
and the kinetic energy expectation value of a normalized 
wave function $\Psi \in \fhn \cap \sobolev^{\otimes N}$
by
\[
\mathcal{E}(\Psi) = \langle\Psi, \, H\Psi\rangle
\quad \text{and} \quad
\mathcal{K}(\Psi) = \bigg\langle\Psi, \, \Big(\sum_{j=1}^N -\Delta_j\Big) \Psi\bigg\rangle \,.
\]
For atoms and molecules the ground state energy $E_{gs}$ is defined as
\begin{align*}
E_{gs} = \inf\Big\{ \mathcal{E}(\Psi) \: \Big| &\:
\Psi \in \fhn \cap \sobolev^{\otimes N} , \ \|\Psi\|_{\fhn} = 1 \ , \\
& \quad R_1,\dots,R_M\in\mathbb R ^3, l\neq m\Rightarrow R_l\neq R_m
\Big\}\,.
\end{align*}
Equipped with this notation, we formulate the coercivity of the energy
functional on the Sobolev space of states with finite kinetic energy:
\begin{prop}
\label{pro:molecule-kinetik-leq-Total-energy}
Consider a neutral atom or a molecule
as in \eqref{eq:Def-Neutral-Atom} or \eqref{eq:Def-Molecule}.
If $E_{gs} \leq 0$ then
\begin{align*}
\mathcal{K}(\Psi) & \leq
\Big( \sqrt{\mathcal{E}(\Psi)-E_{gs}} + \sqrt{-E_{gs}} \Big)^{2}
\leq
2 \mathcal{E}(\Psi) + 4|E_{gs}|\,.
\end{align*}
\end{prop}
\begin{proof}
See \cite[p.132]{MR2583992}.
\end{proof}
Using Proposition~\ref{pro:molecule-kinetik-leq-Total-energy} along
with the conservation of the total energy for both the Schr\"odinger
equation and the TDHF equation (see Appendix~\ref{sec:Theory_TDHF}) we get the following bound on the kinetic
energy.
\begin{prop}
Assume that $\Phi_{HF,0} = \varphi_{1,0} \wedge \cdots \wedge
\varphi_{N,0}$ is a Slater determinant, with $\varphi_{j,0} \in \sobolev$
and $\langle\varphi_{j,0},\varphi_{k,0}\rangle_{\fh}=\delta_{jk}$, for $1
\leq j,k \leq N$. Then, in 
the case of atoms or molecules as in \eqref{eq:Def-Neutral-Atom} or \eqref{eq:Def-Molecule}, 
\begin{align} \label{eq:K-estimate-vb2}
K & \ := \ \sup_{t \geq 0} \Tr[-\Delta \ghf_t] \ \leq \
 \Big( \sqrt{\mathcal{E}(\Phi_{HF,0})-E_{gs}} + \sqrt{-E_{gs}} \Big)^{2} \,.
\end{align}
Thus, if 
$\mathcal{E}(\Phi_{HF,0})\leq0$ then
\begin{align} \label{eq:KKK}
K \ \leq \ -4 E_{gs} \, .
\end{align}
\end{prop}
We also recall a known bound for the ground state energy, see
\cite{PhysRevLett.35.687} or \cite{MR2583992}, whose units we use.
\begin{prop}[Ground state energy of a molecule]
For a molecule with nuclei of charges $Z_1, \ldots, Z_M >0$ at
pairwise distinct positions $R_1, \ldots, R_M \in \mathbb{R}^3$, with
$\lambda=\alpha$, $\nu = \sum_{m<l} \alpha Z_m Z_l / |R_m-R_l|$ as in
\eqref{eq:Def-Molecule}, and $Z=\max\{Z_{1},\dots,Z_{M}\}$, the ground
state energy satisfies the bound
\[
0 \ < \ -E_{gs} \ \leq \ (0.231) \alpha^{2} N 
\bigg[1 + 2.16\, Z \Big( \frac{M}{N} \Big)^{1/3} \bigg]^{2} \,.
\]
\end{prop}

\begin{prop}[Neutral atom]
\label{cor:large-atoms} 
In case of an atom with $N=Z$ the ground state energy satisfies 
\[
0 \ < \ -E_{gs} \ \leq (2.31) \alpha^{2} \, N^{7/3} \,.
\]
\end{prop}

\begin{prop}[TDHF equations without external potential and with repulsive interaction]
\label{cor:mean-field} 
For $h^{(1)}=-\Delta/2$ and $v(x)=|x|^{-1}$, the Hartree-Fock kinetic energy is bounded by the total Hartree-Fock energy (for any $\lambda>0$), which is preserved in time, i.e., 
\[
K \ \leq \ \mathcal{E}(\Phi_{HF,0}) \,.
\]
\end{prop}

Finally, let us note that for attractive Coulomb interaction without external field, we have the bound
\begin{equation}\label{eq:bound_on_e_kin}
K \ \leq \ 2\mathcal{E}(\Phi_{HF,0}) + C \lambda^2 N^{7/3} \, ,
\end{equation}
which follows from the Lieb-Thirring inequality and which we prove in Appendix~\ref{sec:misc_appendix}. Thus, also for attractive interaction, the bounds $K\leq C N$ in the mean-field scaling for large volume \eqref{eq:Def-Mean-Field-Scaling} and $K\leq C N^{5/3}$ in the semi-classical mean-field scaling \eqref{eq:Def-Semi-Classical-Mean-Field} hold, if the corresponding bounds hold for the total energy.

\section{Control of the Degree of Evaporation $S_g$}\label{sec:Relative-Entropy}

We first introduce the degree of evaporation $S_g$, which is a function of a state on the Fock space and a one-particle density matrix. We use $S_g$ as an indicator of closeness of the Hartree-Fock to the Schr\"odinger quantum state.

\subsection{Definition and Properties of the Degree of Evaporation}
For $A$ and more generally $B^{(M)}$ ($M\leq N$) linear operators acting on $\fh$
and $\fhk M$, respectively, we use the notation 
\begin{equation}
\label{eq:def_dGamma_dGammatwo}
\G(A) := \sum_{j=1}^{N} A_{j} 
\qquad\text{and}\qquad
{\G}^{(M)}(B^{(M)}) := \sum_{\substack{j_1,\ldots,j_M=1\\ j_1\neq j_2 \ldots \neq j_M} }^{N} B_{j_1\ldots j_M}^{(M)} \,,
\end{equation}
as operators on $\fhn$, with $A_{j}$ acting on the $j^{th}$ factor in
$\fh^{\otimes N}$ and $B^{(M)}$ acting on the $j_1^{th},\ldots, j_M^{th}$ factors in $\fh^{\otimes N}$, respectively.

\begin{rem}Although we do not use the fermion creation and annihilation operators $a^*$, $a$, note that \eqref{eq:def_dGamma_dGammatwo} coincides with the second quantization $\G$ in quantum field theory in the sense that
\[
\G(A)=\int A(x;y) \ a^*(x)a(y) \ dxdy \, ,
\]
or, more exactly, its restriction to the $N$-particle sector. Similarly, e.g., 
\[
\GG(B^{(2)})=\int B^{(2)}(x_1,x_2;y_1,y_2) \ a^*(x_2)a^*(x_1)a(y_1)a(y_2) \ dx_1dx_2dy_1dy_2 \ .
\]
\end{rem}

Let $\mathcal{L}^{1}(\mathfrak{H})$ denote the space of trace class operators on a Hilbert space $\mathfrak{H}$. We use the partial trace $\Tr_{2}: \mathcal{L}^{1}(\fh^{\otimes2}) \to
\mathcal{L}^{1}(\fh)$ which is defined for 
$B^{(2)}\in\mathcal{L}^{1}(\fh^{\otimes2})$ to be the operator 
$\Tr_{2}[B^{(2)}] \in \mathcal{L}^{1}(\fh)$ such that
\begin{equation}
\Tr\big[ \Tr_{2}(B^{(2)}) \, A \big]
\ = \ 
\Tr\big[ B^{(2)} \, ( A \otimes \Id_{\fh}) \big]
\label{eq:def-partial-trace}
\end{equation}
holds for all $A\in\mathcal{B}(\fh)$.
\begin{defn}\label{one_two_particle_density_matrix}
For an $N$-particle density matrix $\rho \in \mathcal{L}_+^1(\fhn)$,
i.e., a non-negative trace-class operator on $\fhn$ of unit trace, the
one- (resp.\ two-)particle density matrix of $\rho$ is denoted by $\gamma_{\rho}$
(resp. $\gamma_{\rho}^{(2)}$). It is the operator on $\fh$ (resp. $\fhk
2$) such that 
\begin{align} \label{eq:def-1pdm}
\forall A \in \mathcal{B}(\fh): \quad &
\Tr_{\fhn}[\rho\, \G(A)] 
 \ = \
\Tr_{\fh}[\gamma_{\rho} \, A] \,,
\\ \label{eq:def-2pdm}
\forall B^{(2)} \in \mathcal{B}(\fhk 2): \quad & 
\Tr_{\fhn}\big[\rho \, \GG(B^{(2)}) \big] 
 \ = \ 
\Tr_{\fhk 2}\big[ \gamma_{\rho}^{(2)} \, B^{(2)} \big] \,.
\end{align}
\end{defn}
We note that $\gamma_{\rho}$ and $\gamma_{\rho}^{(2)}$ satisfy 
\begin{equation} \label{eq:prop_p-particle_density_matrix}
0 \leq \gamma_{\rho} \leq \Id \,, 
\ \ 
\Tr_{\fh}[\gamma_{\rho}] = N 
\,, \quad
0 \leq \gamma_{\rho}^{(2)} \leq N \,, 
\ \ \Tr_{\fhk 2}[\gamma_{\rho}^{(2)}] = N(N-1) \,
\end{equation}
(see, e.g., \cite[Theorem~5.2]{Bach2013}). Further note that we are
slightly abusing notation since the one-particle density matrix was
defined for wave functions in Eq.~\eqref{eq:p-particle_density_matrix}, rather than for density matrices. We
thus identify $\gamma_\Psi \equiv
\gamma_{|\Psi\rangle\langle\Psi|}$, for all normalized $\Psi \in \fhn$,
whenever this does not lead to confusion.

\begin{defn}
\label{def:S} Let $N\in\mathbb{N}$ and 
\[
\dmn \ := \ 
\big\{ \eta \in \mathcal{L}^{1}(\fh) \, \big| \;
0 \leq \eta \leq \Id \, , \ \Tr[\eta] = N \big\} \,,
\]
and $g$ be a continuous function from $\mathbb R$ to $\mathbb R$. The map $S_g: \mathcal L_+^1(\fhn) \times \dmn \to \mathbb{R}_0^+$ defined by
\begin{align}
S_g(\rho,\eta) 
\ := \
\Tr[\rho \   \hat g]
\,,
\label{eq:Def-S}
\end{align}
where $\hat{g}:=g\big(\G(\Id-\eta )\big)$, is called the degree of evaporation (of $\rho$ relative to $\eta$). The translation of $g$ by $k\in\mathbb{Z}$, is denoted by $\tau_{k}g(y):=g(y-k)$. 
\end{defn}

Note that when the expression of $g$ is too long to fit under a hat, we write $(g)^{\wedge}$ instead of $\hat{g}$.

\begin{rem}
 $\hat{g}:=g\big(\G(\Id-\eta )\big)$ is defined using the functional calculus \cite{MR0493419,MR0071727}. 
\end{rem}
\begin{rem}\label{rem:spectral_decomp}The particular case in which $\eta$ is a rank-$N$ projector is of importance in the sequel, and we then write 
\[
p:=\eta \qquad \text{and} \qquad q:=\Id-\eta \, .
\]
In this case, only the values of $g$ on $\{ 0,\ldots,N \}$ are relevant for the definition of $\hat{g}$ and the continuity assumption on $g$ can be dropped. We can then give an alternative and equivalent viewpoint using the orthogonal projections
\[
P_{m}^{(M)}:=\sum_{\substack{a\in\{0,1\}^{M}_0\\|a|=m}}\bigotimes_{\ell=1}^{M}\big((1-a_{\ell})\ghfnot+a_{\ell}\ghfnotp\big) 
= 1_{\{m\}}\big(\G(\ghfnotp)\big)
\]
on $\fh^{\otimes M}$, with $|a|=a_{1}+\cdots+a_{M}$, for $M\in\mathbb{N}$ and $m\in\mathbb Z$. (Note that with this definition $P_m^{(M)}=0$ for $m\notin \{0,\dots,M\}$.) We can then write down the spectral decomposition of $\G(\ghfnotp)$,
\begin{equation}\label{eq:spectral_decomp}
\G(\ghfnotp) = \sum_{n\in\mathbb Z} n\, P_{n}^{(N)}\,,
\end{equation}
i.e., $P_{n}^{(N)}$ is the projection on the eigenspace of $\G(\ghfnotp)$ associated with the eigenvalue $n \in \mathbb Z$.
It follows that
\begin{equation}
S_g(\rho,\ghfnot) = \sum_{n=0}^N g(n) \, \Tr\big[ \rho P_n^{(N)} \big].
\end{equation}
In this form we see directly that for $\rho = |\Psi\rangle\langle\Psi |$ we have $S_g=\alpha_g$, where $\alpha_g$ is the functional used in \cite[Definition~2.1]{Petrat:2014fk} to control the closeness of a Hartree-Fock state to a Schr\"odinger state. However, note that there is a difference in the choice of normalization. The particular choices of functions, $f$ in \cite{Petrat:2014fk} and $g$ in our article, are related through $g=N\,f$, such that $S_{g}=N\alpha_{f}$. Note also that for $g(x)=x$ we find $S_{Id_{\mathbb R}}(\rho,\ghfnot) = \Tr[ \gamma_{\rho} (\Id-\ghfnot) ]$.
\end{rem}

\begin{rem}
For $g(x)=x$, the functional $S_g$ has been used in \cite{bach:1993,MR1301362,1751-8121-42-8-085201} 
in the context of mean-field approximations for ground states. In \cite{MR1301362}, $S_g$ is called ``degree of non-condensation'' or ``the relative number of particles outside the Fermi sea''. For general $g$, a bosonic variant of $S_g$ was introduced for the derivation of mean-field dynamics in \cite{MR2821235} and for the derivation of the NLS equation in \cite{pickl:2010gp_pos,pickl:2010gp_ext}. For the derivation of mean-field dynamics for fermions, $S_g$ was introduced in \cite{Petrat:2014fk}. Note that for $g(x)=x$ and $\Psi_0\in\fhn$, $2S_g(t)$ coincides with the quantity $\langle\mathcal{U}_{N}(t;0)\xi,
\mathcal{N}\mathcal{U}_{N}(t;0)\xi\rangle$ in \cite{2013arXiv1305.2768B} in case $\xi=R_{\nu_N}^*\Psi_0$.
\end{rem}

Let us collect some properties of $S_g(\rho,\eta)$ and show how it relates to the distance of $\gamma_\rho$ to $\eta$ in trace norm. (Note that some of the statements were already proven in \cite{2013arXiv1305.2768B} and \cite{Petrat:2014fk}.) We denote the Hilbert-Schmidt norm by $\|\cdot\|_{\mathcal{L}^{2}}$.
\begin{prop}
\label{pro:S} For $\eta \in \dmn$ and $\rho$ a density matrix with one particle density matrix $\gamma$, the degree 
$S_g(\rho,\eta)$ of evaporation has the properties

\begin{align} \label{eq:DoE-1}
 \inf_{0 \leq x\leq N} g(x) \leq S_g & (\rho, \eta) \leq  \sup_{0 \leq x\leq N} g(x)\, ,
\\
\label{eq:DoE-2} 
\|\gamma-\eta\|_{\mathcal{L}^{2}}^{2}  \ & \leq \ 
2 S_{Id_{\mathbb R}}(\rho,\eta) \, ,\\
\label{eq:S_positivity_preserving} 
g_1\leq g_2 \text{ on } [ 0,N ] & \Rightarrow S_{g_1}(\rho,\eta)\leq S_{g_2}(\rho,\eta) \, ,\\
\label{eq:S_linear}
g\mapsto S_{g}(\rho,\eta) & \text{ is linear,} 
\end{align}
for $g,g_1,g_2$ functions from $\mathbb R$ to $\mathbb{R}$.

If furthermore $\ghfnot^{2} =\ghfnot$ is a rank-$N$ orthogonal projection and $g(0)=0$, $g(x)\geq x$ on $[0,N]$, then 
\begin{align}\label{eq:DoE-3} 
\frac{1}{N} \| \gamma - \ghfnot \|_{\mathcal{L}^{1}}
\ & \leq \ 
\sqrt{\frac{8}{N}\, S_g(\rho,\ghfnot) } \, ,
\\ \label{eq:S_g_leq_S_Id}
S_g(\rho,\ghfnot) & \leq \sup_{0<x\leq N} \Big|\frac{g(x)}{x}\Big| \| \gamma - \ghfnot \|_{\mathcal{L}^{1}} \, .
\end{align}
\end{prop}

\begin{proof} The spectrum of $\G(\ghfnotp )$ (restricted to $\fhn$) is included in $[0,N]$, thus
\[
 \inf_{0 \leq x\leq N} g(x) \leq \hat{g} \leq  \sup_{0 \leq x\leq N} g(x)\, , 
\]
in the sense of quadratic forms. As $\rho$ is a state, Eq.~\eqref{eq:DoE-1} follows.

Equation~\eqref{eq:DoE-2} follows from 
\begin{align*}
\| \gamma - \eta \|_{\mathcal{L}^{2}}^{2} 
& \ = \ \Tr\big[ (\gamma-\eta)^{2} \big]
\ = \ \Tr\big[ \gamma^{2} + \eta^{2} -2 \gamma\eta \big]
\\[1ex] 
& \ = \ 
2 S_{Id_{\mathbb R}}(\rho,\eta) - \Tr[\gamma-\gamma^2]- \Tr[\eta-\eta^2] 
\ \leq \  
2 S_{Id_{\mathbb R}}(\rho,\eta) \,.
\end{align*}

Equations~\eqref{eq:S_positivity_preserving} and \eqref{eq:S_linear} follow from the properties of the functional calculus.

For the proof of \eqref{eq:DoE-3}, we first remark that $\gamma -
\ghfnot$ has at most $N$ negative eigenvalues (counting
multiplicities). This is a well-known consequence of $\gamma -
\ghfnot \geq - \ghfnot$ and the fact that $\ghfnot$ is a
rank-$N$ orthogonal projection (see, e.g., \cite{MR0493421}), but
we include its proof for the sake of completeness:
Suppose that $\gamma-\ghfnot$ has at least $N+1$ negative
eigenvalues. Then there is a subspace $W$ of dimension $N+1$ such
that $\langle \varphi | (\gamma -
\ghfnot) \varphi \rangle < 0$, for all $\varphi \in W \setminus \{0\}$.
Since $\gamma \geq 0$, this implies that 
$\langle \varphi | \ghfnot \varphi \rangle > 0$, for all 
$\varphi \in W \setminus \{0\}$. On the other hand, the largest
dimension of a subspace with this property is $N$, by the
minmax principle and the fact that $\ghfnot$ has precisely $N$ negative
eigenvalues, which contradicts the existence of $W$.

Denoting the number of negative eigenvalues (counting multiplicities)
of $\gamma-\ghfnot$ by $M$, we consequently have that $M \leq
N$. Let $\lambda_{1},\dots,\lambda_{M}$ be these $M$ negative
eigenvalues of $\gamma-\ghfnot$, and $\lambda_{M+1},
\lambda_{M+2}, \dots$ be the non-negative ones. Since
$\Tr[\gamma-\ghfnot] = 0$, it follows that 
\[
-(\lambda_{1}+\dots+\lambda_{M})
\ = \
\sum_{m=M+1}^{\infty}\lambda_{m} \,.
\]
Using the Cauchy-Schwarz inequality and $M\leq N$, we obtain
\begin{align*}
\|\gamma - \ghfnot\|_{\mathcal{L}^{1}} 
& \ = \
\sum_{m=M+1}^{\infty} \lambda_{m} - \sum_{m=1}^{M} \lambda_{m}
\ = \ 
- 2 \sum_{m=1}^{M}\lambda_{m}
\ \leq \ 
2 \sqrt{M} \bigg( \sum_{m=1}^{M} \lambda_{m}^{2} \bigg)^{1/2}
\\ &
\ \leq \
2 \sqrt{N} \bigg( \sum_{m=1}^{\infty} \lambda_{m}^{2} \bigg)^{1/2}
\ = \ 
2 \sqrt{N} \|\gamma-\ghfnot\|_{\mathcal{L}^{2}} \,,
\end{align*}
and Eq.~\eqref{eq:DoE-3} follows from Eq.~\eqref{eq:DoE-2} and $S_{Id_\mathbb R}(\rho,\ghfnot) \leq S_{g}(\rho,\ghfnot) $. 
  
To prove Eq.~\eqref{eq:S_g_leq_S_Id}, we observe that $g(x)\leq \sup_{0<x\leq N}\Big\{\Big|\frac{g(x)}{x}\Big|\Big\} x$  on $[0,N]$ and thus, using positivity preservation and linearity, we have
\[
S_{g}(\rho,\ghfnot) \leq \sup_{0<x\leq N}\Big\{\Big|\frac{g(x)}{x}\Big| x\Big\} S_{Id_\mathbb R}(\rho,\ghfnot)\ .
\]
We conclude with
\[
S_{Id_\mathbb R}(\rho,\ghfnot) 
\ = \ 
\Tr[ \gamma (1 - \ghfnot)]
\ = \ 
\Tr[ \ghfnot (\ghfnot - \gamma) \ghfnot ]
\ \leq \ 
\| \gamma - \ghfnot \|_{\mathcal{L}^{1}} \, , 
\]
using again $\ghfnot = \ghfnot^2$ and $\Tr[\gamma]=N=\Tr[\ghfnot]$.
\end{proof}

Let us now state the main result of this section. Recall that we defined $K := \sup_{t \geq 0} \Tr[-\Delta \ghf_t]$.

\begin{thm}\label{thm:estimate-S} 
Assume \eqref{eq:K-uniformly-bounded} holds, i.e., the kinetic energy is uniformly bounded, as in  Theorem~\ref{thm:estimate-trace-norm}. Then, writing $\rho_t = |\Psi_{t}\rangle\langle\Psi_{t}|$,
\begin{align}\label{eq:estimation_S}
S_{g_{1/3}}(\rho_t,\ghft) \leq S_{g_{1/3}}(\rho_0,\ghfnot_0) \exp\big(30\lambda \sqrt{K} N^{1/6} t\big)+N^{2/3} \Big( \exp\big(30\lambda \sqrt{K} N^{1/6} t\big) - 1 \Big) \, , \end{align}
where, for $\theta >0$, $g_\theta$ is the function from $\mathbb R$ to $\mathbb R$ defined by
\begin{equation}\label{eq:g_theta}
 \forall x\in\mathbb{R}\,, \quad
g_{\theta}(x):=N^{1-\theta} \ x \ 1_{[0,N^{\theta}]}(x)+ N\  1_{(N^{\theta},\infty)}(x) \,. 
\end{equation}
\end{thm}
Note that the function $g_{\theta}$ was also used to obtain the results in \cite{Petrat:2014fk}. Theorem~\ref{thm:estimate-S} will be proved in the following subsections. The strategy is to obtain a bound for $dS_t/dt$ in terms of $S_t$ and $N^{\delta}$, for some $\delta<1$, and then integrate it, in the spirit of the Gr\"onwall lemma.

Before we turn to the proof of Theorem~\ref{thm:estimate-S}, we show how Theorem~\ref{thm:estimate-S} and the properties of the degree of evaporation 
imply Theorem~\ref{thm:estimate-trace-norm}, the main result of this article.
\begin{proof}[Proof of Theorem~\ref{thm:estimate-trace-norm}] 
Since $g_{1/3} \geq Id_{\mathbb R}$ on $[0,N]$, we can apply Eq.~\eqref{eq:DoE-3} to Eq.~\eqref{eq:estimation_S} which gives
\begin{align*}
\frac{1}{N} \| \gamma_t - \ghft \|_{\mathcal{L}^{1}} & \leq \sqrt{\frac{8}{N}}\sqrt{S_{g_{1/3}}(\rho_0,\ghfnot_0) \exp (30\lambda \sqrt{K} N^{1/6} t)+N^{2/3} \big( \exp (30\lambda \sqrt{K} N^{1/6} t) - 1 \big) } \, .
\end{align*}
Equation~\eqref{eq:S_g_leq_S_Id} with $g_{1/3}$ yields $S_{g_{1/3}}(\rho,\ghfnot)  \leq N^{2/3}  \| \gamma - \ghfnot \|_{\mathcal{L}^{1}}$ which then gives Eq.~\eqref{eq:main_result}.
\end{proof}

The rest of this section is devoted to the proof of Theorem~\ref{thm:estimate-S}.

\subsection{\label{sec:time-derivative-S}Time-Derivative of the Degree of Evaporation}

In this subsection we calculate the time derivative of the degree of evaporation $S_{g}(t) := S_g(\rho_t,\ghf_t)$ and bring it into a form that can be conveniently estimated. Then most of the following subsections provide bounds on the different contributions to the time derivative. First, recall the Fefferman-de la Llave decomposition, for $x\neq y\in\mathbb R^3$,
\begin{equation} \label{eq:Feffermann-de_la_Llave1}
\frac{1}{|x-y|}
\ = \ 
\int_{\mathbb{R}^3} d^3z \int_0^\infty \frac{dr}{\pi \, r^{5}} \:
X_{r,z}(x) \, X_{r,z}(y) \, ,
\end{equation} 
of the Coulomb potential, where $X_{r,z}(x) := 1_{|x-z|\leq r}$ is the
characteristic function of the ball in $\mathbb{R}^3$ of radius $r >0$
centered at $z \in \mathbb{R}^3$. This formula can also be written as 
\begin{equation} \label{eq:Feffermann-de_la_Lave2}
v^{(2)} \ = \ \int d\mu(\om) \: X_\om \otimes X_\om \, ,
\end{equation}
where $\om = (r,z) \in \mathbb{R}^+ \times \mathbb{R}^3$ and $\int d\mu(\om) \, f(\om) := \int_{\mathbb{R}^3} d^3z \int_0^\infty
\frac{dr}{\pi \, r^{5}} \: f(r,z)$.  The form~\eqref{eq:Feffermann-de_la_Llave1} is convenient for the estimates
derived below, but we note that it agrees with
\eqref{eq:Hainzl-Seiringer}, of course. 
Note that the terms $\mathcal{A}_t$, $\mathcal{B}_t$, $\mathcal{C}_t$ in the following proposition are the same as  $(I)$, $(II)$, $(III)$ in \cite[Lemma~6.5]{Petrat:2014fk}. However, an important difference lies in the presentation of the $\mathcal{A}_t$ term using the  decomposition~\eqref{eq:Feffermann-de_la_Llave1}, which enables us to handle the case of the Coulomb interaction.
In the following, $\Im$ denotes the imaginary part.
\begin{prop}\label{pro:equation-dS_g/dt} For all monotonically increasing $g:\mathbb{R}\to\mathbb{R}$, the time-derivative of 
$S_{g}(t) = S_g(\rho_t,\ghf_t)$ (with the notation from Theorems~\ref{thm:estimate-trace-norm} and \ref{thm:estimate-S}) is
\begin{equation} \label{eq:dS/dt}
\frac{dS_{g}(t)}{dt} = \lambda \big( \mathcal{A}_t + \mathcal{B}_t + \mathcal{C}_t \big) \, ,
\end{equation}
where
\begin{align} \label{eq:A_t}
\mathcal{A}_t & 
:= \int 2\Im \Tr\big[ \G(\ghftp X_\om\ghft)\big(\G(\ghft X_\om\ghft)-\Tr[X_\om\ghft]\big)\rho_t^{[-1,1]} \big] d\mu(\omega) \, ,
\\[1ex] \label{eq:b(X)}
 \mathcal{B}_t & := \Im \Tr\big[\GG\big((\ghftp \otimes\ghftp )v^{(2)}(\ghft\otimes\ghft)\big)\,\rho_t^{[-2,2]}\big]  \, ,
\\[1ex] \label{eq:c}
\mathcal{C}_t & := 2\Im \Tr\big[\GG\big((\ghftp \otimes\ghftp )v^{(2)}(\ghft\otimes\ghftp )\big)\, \rho_t^{[-1,1]}\big] \,,
\end{align}
with $\ghftp := \Id-\ghf_t$, and 
\begin{equation}
 \rho^{[-j,j]} :=(\tau_{-j}g-g)^{1/2\wedge}\,\rho (g-\tau_{j}g)^{1/2\wedge}\ .
\end{equation}
\end{prop}
\noindent
Before we turn to the proof we note that 
\begin{equation} \label{eq:orthogonality-gamma_HF}
  \ghftp \, \ghf_t \ = \  \ghf_t \, \ghftp \ = \ 0 
 \, ,
\end{equation}
since $\ghf_t$ is a projection. We further note that, for
$A$, $B$ linear and bounded operators on $\fh$, we have that
\begin{equation} \label{eq:dGdG_dG2+dG}
\G(A) \, \G(B) 
\ = \ 
\GG (A \otimes B) + \G(AB) \,.
\end{equation}
To prove Proposition~\ref{pro:equation-dS_g/dt} we need several lemmas. We begin with an evolution equation for $\hat g$.
\begin{lem}\label{lem:evolution_equation_g_hat}
For any function $g:\mathbb R \to \mathbb R$, with $\hat g=g\big(\G(\ghftp )\big)$ and $h^{(1)}_{HF}$ defined in Eq.~\eqref{eq:hartree-fock_HF-Hamiltonian}, 
\begin{equation}\label{eq:evolution_hat_g}
i\partial_{t}\hat{g}=[\G\big(h^{(1)}_{HF}(\ghft)\big),\hat{g}]\,.
\end{equation}

\end{lem}
\begin{proof}
First observe that only the values of $g$ on the spectrum of $\G (\ghftp)$ are used to define $\hat g$. Hence we could as well consider a polynomial which coincides with $g$ on the spectrum of $\G (\ghftp)$. It is then enough to prove that Eq.~\eqref{eq:evolution_hat_g} holds for any monomial $\G(\ghftp )^n$. It indeed holds for $n=1$:
\[
i\partial_{t}\G(\ghftp )=\G(i\partial_{t}\ghftp )=\G\big([h^{(1)}_{HF}(\ghft),\ghftp ]\big)=\big[\G\big(h^{(1)}_{HF}(\ghft)\big),\G(\ghftp )\big]\,.
\]
Then, for any $n\in \mathbb N$,
\begin{align*}
i\partial_{t} \big( \G(\ghftp ) \big)^n &=
\sum_{j=1}^{n} \big( \G(\ghftp ) \big)^{j-1} \
\big[\G\big(h^{(1)}_{HF}(\ghft)\big),\G(\ghftp )\big] \
 \big( \G(\ghftp ) \big)^{n-j} \\
 &=\big[\G\big(h^{(1)}_{HF}(\ghft)\big),\big( \G(\ghftp ) \big)^n\big]
\,,
\end{align*}
as all the terms but two simplify in the sum.
\end{proof}
Next, we need a commutation relation (analogous to \cite[Lemma~6.4]{Petrat:2014fk}) involving $\hat g$ that enables us to write the time derivative of $S_g(t)$ in terms of a discrete derivative of $g$.
Recall that $\tau_{j-k}g(x)=g(x-j+k)$.
\begin{lem}
\label{pro:dGamma_theta__tau_theta_dGamma}
For integers $0\leq j,k\leq M\leq N$, any function $g:\mathbb R \to \mathbb R$, with the notations of Definition~\ref{def:S}, and $h^{(M)}\in\mathcal L(\fh^{\otimes M})$,
\begin{equation}
\label{eq:dGamma_theta__tau_theta_dGamma}
{\G}^{(M)}\big(P_{j}^{(M)}h^{(M)}P_{k}^{(M)}\big)\,\hat{g}=\widehat{\tau_{j-k}g}\, {\G}^{(M)}\big(P_{j}^{(M)}h^{(M)}P_{k}^{(M)}\big)\,.
\end{equation}
\end{lem}
\begin{proof}
First, note that if
\[
{\G}^{(M)}\big(P_{j}^{(M)}h^{(M)}P_{k}^{(M)}\big)\G(\ghfnotp)=\widehat{\tau_{j-k}Id}\: {\G}^{(M)}\big(P_{j}^{(M)}h^{(M)}P_{k}^{(M)}\big)
\]
holds, then Eq.~\eqref{eq:dGamma_theta__tau_theta_dGamma} follows by 
the same argument as in the proof of Lemma \ref{lem:evolution_equation_g_hat}.
Using  $P_{m_{1}}^{(M)}P_{m_{2}}^{(M)}=\delta_{m_{1}m_{2}}P_{m_{1}}^{(M)}$
and $P_{m}^{(N)}=\sum_{d\in\mathbb Z}P_{d}^{(M)}\otimes P_{m-d}^{(N-M)}$ (recall that $P_d^{(M)}=0$ for $d\notin\{0,\dots,M\}$),
and without loss of generality singling out the first $M$ variables,
\begin{align*}
\Big(\big(P_{j}^{(M)}h^{(M)}P_{k}^{(M)}\big)\otimes \Id ^{\otimes N-M}\Big)P_{n}^{(N)} & =\Big(\big(P_{j}^{(M)}h^{(M)}P_{k}^{(M)}\big)\otimes \Id^{\otimes N-M}\Big)\Big(\sum_{d\in\mathbb{Z}}P_{d}^{(M)}\otimes P_{n-d}^{(N-M)}\Big)\\
 & =\big(P_{j}^{(M)}h^{(M)}P_{k}^{(M)}\big)\otimes P_{n-k}^{(N-M)}\\
 & =\Big(\sum_{d\in\mathbb{Z}}P_{d}^{(M)}\otimes P_{n-k+j-d}^{(N-M)}\Big)\Big(\big(P_{j}^{(M)}h^{(M)}P_{k}^{(M)}\big)\otimes \Id^{\otimes N-M}\Big)\\
 & =P_{n-k+j}^{(N)}\Big(\big(P_{j}^{(M)}h^{(M)}P_{k}^{(M)}\big)\otimes \Id^{\otimes N-M}\Big)\ .\end{align*}
It follows from the spectral decomposition \eqref{eq:spectral_decomp} of $\G(\ghfnotp)$ 
that 
\begin{align*}
{\G}^{(M)}\big(P_{j}^{(M)}h^{(M)}P_{k}^{(M)}\big)\G(\ghfnotp) & ={\G}^{(M)}\big(P_{j}^{(M)}h^{(M)}P_{k}^{(M)}\big) \sum_{n\in\mathbb{Z}}n\, P_{n}^{(N)}\\
 & =\sum_{n\in\mathbb{Z}}n\, P_{n-k+j}^{(N)}{\G}^{(M)}\big(P_{j}^{(M)}h^{(M)}P_{k}^{(M)}\big)\\
 & =\sum_{n\in\mathbb{Z}}(n+k-j)\, P_{n}^{(N)}{\G}^{(M)}\big(P_{j}^{(M)}h^{(M)}P_{k}^{(M)}\big)\\
 & =\widehat{\tau_{j-k}Id}\: {\G}^{(M)}\big(P_{j}^{(M)}h^{(M)}P_{k}^{(M)}\big) \, ,
\end{align*}
which, as discussed above, implies the result.
\end{proof}
\begin{proof}[Proof of Proposition \ref{pro:equation-dS_g/dt}]Without loss of generality we assume that $\Psi_0 \in \fhn \cap \sobolev^{\otimes N}$ (see Remark~\ref{rem:proof_in_H1_is_enough}). Using the evolution equation for $\hat g$ from Lemma~\ref{lem:evolution_equation_g_hat}, we find
\begin{align*}
\frac{dS_{g}}{dt}(\rho_t,\ghft) & =\Tr\big[-i[H,\rho_t]\hat{g}+\rho_t\,(-i[\G(h_{HF}^{(1)}(\ghft)),\hat{g}])\big]\\
 & =\Tr\big[i[ H - \G(h_{HF}^{(1)}(\ghft)),\hat{g}]\rho_t\big]\\
 & =\lambda \Tr\big[i[\frac{1}{2}\GG(v^{(2)})-\G(v_{HF}^{(1)}(\ghft)),\hat{g}]\rho_t\big]
 \end{align*}
with $v^{(1)}_{HF}(\ghft):=\Tr_2[v^{(2)}(1-\mathfrak X)(\Id\otimes\ghft)]$. Now, recall that $\sum_{m=0}^M P_m^{(M)}=\Id_{\fh^{\otimes M}}$.
Inserting this identity for $M=1$ and $M=2$ and using Lemma~\ref{pro:dGamma_theta__tau_theta_dGamma} gives
\begin{align*}
\frac{dS_{g}}{dt}(\rho_t,\ghft) & =\frac{\lambda}{2} \Tr\big[i[\GG(\big(P_{0}^{(2)}+P_{1}^{(2)}+P_{2}^{(2)}\big)v^{(2)}\big(P_{0}^{(2)}+P_{1}^{(2)}+P_{2}^{(2)}\big))\allowdisplaybreaks\\
 &\qquad - 2\G(\big(P_{0}^{(1)}+P_{1}^{(1)}\big)v^{(1)}_{HF}(\ghft)\big(P_{0}^{(1)}+P_{1}^{(1)}\big)),\hat{g}]\rho_t\big]\allowdisplaybreaks\\
 & =\frac{\lambda}{2} \Tr\big[i[\GG(P_{0}^{(2)}v^{(2)}P_{2}^{(2)}+P_{2}^{(2)}v^{(2)}P_{0}^{(2)}\\
 & \qquad+P_{0}^{(2)}v^{(2)}P_{1}^{(2)}+P_{1}^{(2)}v^{(2)}P_{0}^{(2)}\\
 & \qquad+P_{1}^{(2)}v^{(2)}P_{2}^{(2)}+P_{2}^{(2)}v^{(2)}P_{1}^{(2)})\\
 &\qquad -2\G\big(P_{0}^{(1)}v^{(1)}_{HF}(\ghft) P_{1}^{(1)}+P_{1}^{(1)}v^{(1)}_{HF}(\ghft) P_{0}^{(1)}\big),\hat{g}]\rho_t\big]\allowdisplaybreaks\\
& = \lambda\Im \Tr\big[\big(\GG(P_{1}^{(2)}v^{(2)}P_{0}^{(2)})
-2\G\big(P_{1}^{(1)}v^{(1)}_{HF}(\ghft) P_{0}^{(1)}\big)\big)\,\rho_t^{[-1,1]}\big]\\
& \quad+\lambda\Im \Tr\big[\GG(P_{2}^{(2)}v^{(2)}P_{0}^{(2)})\,\rho_t^{[-2,2]}\big]\\
& \quad+\lambda\Im \Tr\big[\GG(P_{2}^{(2)}v^{(2)}P_{1}^{(2)})\,\rho_t^{[-1,1]}\big]\ .
 \end{align*}
We then insert the Fefferman-de la Llave decomposition and Eq.~\eqref{eq:dGdG_dG2+dG} in the first term to get
\begin{align*}
\Im & \Tr[(\GG(P_{1}^{(2)}v^{(2)}P_{0}^{(2)})-2\G(P_{1}^{(1)}v_{HF}^{(1)}(\ghft)P_{0}^{(1)}))\rho_t^{[-1,1]}]\\
 & =2\Im\int \Tr\big[\big(\GG(\ghftp X_{\omega}\ghft\otimes\ghft X_{\omega}\ghft)\\
 & \qquad-\Tr[X_{\omega}\ghft]\G(\ghftp X_{\omega}\ghft)+\G(\ghftp X_{\omega}\ghft X_{\omega}\ghft)\big) \rho_t^{[-1,1]} \big]d\mu(\omega)\\
 & =2\Im\int \Tr\big[ \G(\ghftp X_{\omega}\ghft)\big(\G(\ghft X_{\omega}\ghft)-\Tr[X_{\omega}\ghft])\big)\rho_t^{[-1,1]} \big]d\mu(\omega)\,,
\end{align*}
where we used $P_0^{(1)}=\ghft,\,P_1^{(1)}=\ghftp$, $P_0^{(2)}=\ghft^{\otimes 2},\,P_1^{(2)}=\ghftp\otimes\ghft\,+\,\ghft\otimes\ghftp$.
\end{proof}

\subsection{Auxiliary Lemmas}
We prove here three lemmas that we frequently need for estimating the terms $\mathcal{A}_t$, $\mathcal{B}_t$ and $\mathcal{C}_t$ from Proposition~\ref{pro:equation-dS_g/dt}.

For fermionic systems the following bound on $\G(A)$ is well-known. Note that this is the only point at which the Fermi statistics enter our paper.
\begin{lem} \label{lem:dGamma-leq-Tr}
Let $A$ be a trace-class and self-adjoint operator on a separable Hilbert space $\fh$. Then, as quadratic forms on $\fhn$,
\[
\G(A) \:
\leq \: \|A\|_{\mathcal{L}^{1}}
\,.
\] 
\end{lem}
\begin{proof}
We use the spectral decomposition $A=\sum_{j}\lambda_{j}|\varphi_{j}\rangle\langle\varphi_{j}|$ with $\lambda_{j}\in\mathbb{R}$, $\sum_{j}|\lambda_{j}|<\infty$, for some orthonormal basis $(\varphi_{j})_{j=1}^{\infty}$, and we write any vector $\Psi\in\fhn$ as
\[
\Psi=\sum_{j_{1}<\cdots<j_{N}}\alpha_{j_{1},\dots,j_{N}}\varphi_{j_{1}}\wedge\cdots\wedge\varphi_{j_{N}} \, ,
\]
where $\|\Psi\|^{2}=\sum|\alpha_{j_{1},\dots,j_{N}}|^{2}<\infty$. Then
\begin{align*}
\langle\Psi,\G(A)\Psi\rangle & =\langle\Psi,\sum_{j_{1}<\cdots<j_{N}}(\lambda_{j_{1}}+\dots+\lambda_{j_{N}})\alpha_{j_{1},\dots,j_{N}}\varphi_{j_{1}}\wedge\cdots\wedge\varphi_{j_{N}}\rangle\allowdisplaybreaks\\
 & =\sum_{j_{1}<\cdots<j_{N}}(\lambda_{j_{1}}+\dots+\lambda_{j_{N}})|\alpha_{j_{1},\dots,j_{N}}|^{2}\allowdisplaybreaks\\
 & \leq\sum_{j_{1}<\cdots<j_{N}}\|A\|_{\mathcal{L}^{1}}\,|\alpha_{j_{1},\dots,j_{N}}|^{2}=\|A\|_{\mathcal{L}^{1}}\,\|\Psi\|^{2}\,,\end{align*}
which yields the result.
\end{proof}

We recall the definition of the direct integral of a family $(B(x))_{x\in\mathbb{R}^3}$ of operators on~$\fh$:
\begin{equation}
\Big[\Big(\int_{\mathbb{R}^3}^{\oplus}B(x^\prime_1)dx^\prime_1\Big)
\psi\Big](x_1,x_2):=\big[\big(\Id\otimes B(x_1)\big)\psi\big] (x_1,x_2)\,,
\end{equation}
for any $\psi\in \fhk{2}$. 
\begin{lem}\label{lem:dGamma2_w}
Let $A$ be a bounded non-negative operator on $\fh$ and $(B(x))_{x\in\mathbb{R}^3}$ be a family of non-negative trace class operators on $\fh$. 
Then
\begin{align}
\label{eq:estimate_direct_integral1}
\GG\Big(( \sqrt A\otimes \Id)\int^{\oplus}_{\mathbb{R}^3} B(x_1)dx_1( \sqrt A\otimes \Id)\Big) 
&\leq \G\big( \sqrt A \, \Tr[B(x)] \,  \sqrt A\big) \,.
\end{align}
If $A$ is also trace -class and such that $\Tr\big[\sqrt{A}\Tr[B(x)]\sqrt{A}\big]<\infty$ then
\begin{align}
\GG\Big(( \sqrt A\otimes \Id)\int^{\oplus}_{\mathbb{R}^3} B(x_1)dx_1( \sqrt A\otimes \Id)\Big) 
\label{eq:estimate_direct_integral2}
&\leq \int A(x;x) \,  \Tr[B(x)] dx \,.
\end{align}
where $A(x,y)=\sum_{i=1}^\infty \lambda_i \varphi_i (x) \overline{\varphi_i (y)}$ denotes the integral kernel of $A$ defined in terms of the spectral decomposition of $A$.

In particular:

If $B(x)=B$ does not depend on $x$,
\begin{equation}
\GG( A\otimes B) 
\leq \Tr[B] \G( A )\,.
\end{equation}
With $w:\mathbb{R}^{3}\to\mathbb{R}^{+}$, $w^{(2)}=w(x_{1}-x_{2})$
and $\ghfnot=\sum_{i=1}^N |\varphi_i\rangle\langle \varphi_i|$, $\langle\varphi_i|\varphi_j\rangle=\delta_{ij}$ a rank-$N$ projector on $\mathfrak{h}$, we have that
\begin{align}\label{eq:dGamma2_w_eta_purp}
\GG\big((\ghfnotp\otimes\ghfnot)w^{(2)}(\ghfnotp\otimes\ghfnot)\big)
& \leq \G\big(\ghfnotp(w*f)\ghfnotp\big)
 \leq \|w*f\|_\infty \ \G(\ghfnotp)\,,
\\
\label{eq:dGamma2_w_eta}
\GG\big((\ghfnot\otimes\ghfnot)w^{(2)}(\ghfnot\otimes\ghfnot)\big)
&\leq \G\big(\ghfnot(w*f)\ghfnot\big)
\leq \|(w*f)f\|_1\,,
\end{align}
where $f(x):=\ghfnot(x;x):=\sum_{i=1}^N |\varphi_i(x)|^2$ are the diagonal values of the integral kernel of~$\ghfnot$.
\end{lem}

\begin{proof}
 Let $\Psi^{(N)}\in\fhn$. With
\[
\tilde{\Psi}_{A,x}^{(N-1)}(x_{1},\dots,x_{N-1}):=\Big(\big( \Id^{\otimes N-1}\otimes \sqrt A\big)\Psi^{(N)}\Big)(x_{1},\dots,x_{N-1},x)\]
and the direct integral representation we get, using Lemma~\ref{lem:dGamma-leq-Tr},
\begin{align*}
\big\langle\Psi^{(N)}, &\GG\Big((\sqrt  A\otimes \Id)\int^{\oplus}_{\mathbb{R}^3} B(x_1)dx_1( \sqrt A\otimes \Id)\Big) \Psi^{(N)}\big\rangle\\
 & =N(N-1)\int\Big\langle\tilde{\Psi}_{A,x}^{(N-1)}, \big(\Id^{\otimes N-2} \otimes B(x)  \big)\tilde{\Psi}_{A,x}^{(N-1)}\Big\rangle \, dx\\
 & =N\int\Big\langle\tilde{\Psi}_{A,x}^{(N-1)},\G\big(B(x)\big)\tilde{\Psi}_{A,x}^{(N-1)}\Big\rangle \, dx \\
 & \leq N\int\Big\langle\tilde{\Psi}_{A,x}^{(N-1)},\Tr[B(x)]\tilde{\Psi}_{A,x}^{(N-1)}\Big\rangle \, dx\\
 & =\big\langle\Psi^{(N)},\G\big(\sqrt A\Tr[B(x)]\sqrt A\big)\Psi^{(N)}\big\rangle \,,
\end{align*}
and Eq.~\eqref{eq:estimate_direct_integral1} follows. If $\Tr[\sqrt A \Tr[B(x)]\sqrt A]<\infty$, then Eq.~\eqref{eq:estimate_direct_integral2} follows from Lemma \ref{lem:dGamma-leq-Tr}.

The case with $B(x)$ independent of $x$ is clear.

For the second particular case, observe that the operator $(\Id\otimes\ghfnot)w^{(2)}(\Id \otimes\ghfnot)$ can be
written as the direct integral \[
(\Id\otimes\ghfnot)w^{(2)}(\Id\otimes\ghfnot)=\int_{\mathbb{R}^{3}}^{\oplus}\ghfnot(\tau_{x_1}w)\ghfnot\, dx_1\]
with $\tau_{x_1}w(x_2)=w(x_2-x_1)$ a translation of $w$. Then with  $B(x)=\ghfnot (\tau_x w) \ghfnot$ and $A=\ghfnotp$ we get Eq.~\eqref{eq:dGamma2_w_eta_purp}, and with $A=\ghfnot$ we get Eq.~\eqref{eq:dGamma2_w_eta}.
\end{proof}

For $\rho \in \mathcal{L}^1(\fhn)$, let us introduce the shorthand notation
\begin{align}
\rhoj &:= (g-\tau_{j}g)^{1/2\,\wedge} \ \rho \ (g-\tau_{j}g)^{1/2\,\wedge} \,, \label{eq:def_rhoj} \\
\rhomj &:= (\tau_{-j}g-g)^{1/2\,\wedge} \ \rho \ (\tau_{-j}g-g)^{1/2\,\wedge} \,, \label{eq:def_rhomj}
\end{align}
with $j=1,2$, and let $\gammaj$ and $\gammamj$ be the corresponding one-particle and $\gamma^{[j](k)}$ and $\gamma^{[-j](k)}$ the corresponding $k$-particle density matrices (see also Definition~\ref{one_two_particle_density_matrix} extended to non-negative and trace class operators whose trace is not necessarily one). Note that $\rhoj$ and $\rhomj$ are not states because their trace is not one, and thus $\gammaj$ and $\gammamj$ do not necessarily satisfy Eq.~\eqref{eq:prop_p-particle_density_matrix}.

The next lemma shows the advantage we gain from using the function
\begin{equation}\label{eq:g_theta_again}
g_{\theta}(x):=N^{1-\theta} \ x \ 1_{[0,N^{\theta}]}(x)+ N\  1_{(N^{\theta},\infty)}(x) \,
\end{equation}
in the definition of the degree of evaporation. This lemma is analogous to \cite[Lemma~7.1]{Petrat:2014fk}, but note that the use of the functional calculus clarifies the fact that one ultimately uses only inequalities on functions from $\mathbb{R}$ to $\mathbb{R}$.

\begin{lem} \label{lem:estimate_by_S_theta}
For $j\in\{-2,-1,1,2\}$, any (normalized) state $\rho\in\mathcal L^1\big(\fhn\big)$, and the function $g_{\theta}$ from \eqref{eq:g_theta_again} (with the notation from \eqref{eq:def_rhoj} and \eqref{eq:def_rhomj}), 
\begin{align*}
\Tr \big[\rhoj\big]  & \leq |j| \ N^{1-\theta}\,,\\
\Tr \big[ \G(\ghfnotp) \ \rhoj \big] & \leq |j|\,(|j|+1) \ S_{g_{\theta}}\,,\\
\Tr \big[\GG(\ghfnotp\otimes\ghfnotp) \ \rhoj \big] & \leq |j|\,(|j|+1)^{2} \ N^{\theta} \ S_{g_{\theta}}\,.
\end{align*}
\end{lem}

\begin{proof}[Proof of Lemma~\ref{lem:estimate_by_S_theta}]
The inequalities are a direct consequence of the functional calculus, once we observe that 
$\G(\ghfnotp)=\widehat{Id_{\mathbb{R}}}$, 
$\GG(\ghfnotp\otimes\ghfnotp)=\G(\ghfnotp)^{2}-\G(\ghfnotp)=(Id_{\mathbb{R}}\cdot(Id_{\mathbb{R}}-1))^{\wedge}\:$ 
and
\begin{align*}
\tau_{j}g_{\theta}-\tau_{k}g_{\theta} & \leq(k-j) \ N^{1-\theta}\,,\\
Id_{\mathbb{R}}\cdot(\tau_{j}g_{\theta}-\tau_{k}g_{\theta}) & \leq(k-j)(k-j+1) \ g_{\theta}\,,\\
Id_{\mathbb{R}}\cdot(Id_{\mathbb{R}}-1)\cdot(\tau_{j}g_{\theta}-\tau_{k}g_{\theta}) & \leq(k-j)\,(k-j+1)^{2} \ N^\theta \ g_{\theta}\,,\end{align*}
as inequalities of functions from $\mathbb R$ to $\mathbb R$, for $-2\leq j<k\leq 2$. 
\end{proof}

\subsection{Bound for $\mathcal{A}_t$}

Let us first estimate the integrand $\mathcal{A}_t(X_\om)$ of \eqref{eq:A_t}, i.e.,
\[
\mathcal{A}_t(X) := 2\Im \Tr\big[ \G(\ghftp X\ghft)\big(\G(\ghft X\ghft)-\Tr[X\ghft]\big)\rho_t^{[-1,1]} \big] \, ,
\]
where $X$ is an operator on $\fh$ such that $0\leq X\leq \Id$.
\begin{prop}
\label{pro:estimate_a(X)} Let $X$ be an operator on $\fh$
such that $0\leq X\leq \Id$ and set ${ \gammamone_t }^{\perp}:= \Tr[ \rhomone_t] - \gammamone_t$.  
Then
\begin{align} \label{eq:a(X)_leq_blabla}
\mathcal A_t(X) & \ \leq \ 
\Tr[\ghf_t X] \, \Tr[ X (2 \ghftp \gammaone_t \ghftp 
+ \ghf_t {\gammamone_t}^{\perp} \ghf_t) ] \, .
\end{align}
\end{prop}

\begin{proof}
Using the Cauchy-Schwarz inequality and $2ab\leq a^{2}+b^{2}$, and then Eq.~\eqref{eq:dGdG_dG2+dG} and Lemma~\ref{lem:dGamma-leq-Tr}, we get
\begin{align} \label{eq:a(X)_first_estimate}
\mathcal A_t(X) \ = \ & 2\Im \Tr\big[ \G(\ghftp X\ghft)\big(\G(\ghft X\ghft)-\Tr[X\ghft]\big)\rho_t^{[-1,1]} \big]  \nonumber \\[1ex]
\ \leq \ &
\Tr\big[ \G(\ghftp X \ghf_t) \, 
\G(\ghf_t X \ghftp) \, \rhoone_t \big] 
\: + \: 
\Tr\big[ \big( \Tr[\ghf_tX] - \G(\ghf_t X \ghf_t) \big)^{2} \rhomone_t  \big]
\nonumber \\[1ex] 
\ \leq \ &
\Tr\big[ \GG(\ghftp X \ghf_t \otimes \ghf_t X \ghftp)
\, \rhoone_t \big] \: + \: 
\Tr\big[ \G(\ghftp X \ghf_t^{2} X \ghftp) \, \rhoone_t \big]
\nonumber \\ &
\: + \: \Tr[\ghf_tX] \, 
\Tr\big[ \big(\Tr[\ghf_tX] - \G(\ghf_t X \ghf_t) \big) \, \rhomone_t \big] \, .
\end{align}
For the first term on the right-hand side of
\eqref{eq:a(X)_first_estimate}, we apply the Cauchy-Schwarz inequality
again and obtain 
\begin{align} \label{eq:gpXg-gXgp-leq-gpXgp-gXg}
\Tr\big[ \GG(& \ghftp X \ghf_t  
\otimes \ghf_t X \ghftp)
\, \rhoone_t \big]
\ = \ 
\Tr\Big[ \big( \ghftp \sqrt{X} \otimes \ghf_t \sqrt{X} \big)
\, \big( \sqrt{X} \ghf_t \otimes \sqrt{X} \ghftp \big) 
\, \gamma_t^{[1](2)} \Big]
\nonumber \\[1ex] 
& \ \leq \
\sqrt{ \Tr\big[ (\ghftp X \ghftp \otimes \ghf_t X \ghf_t) \, 
\gamma_t^{[1](2)} \big] \:} \: 
\sqrt{ \Tr\big[ (\ghf_t X \ghf_t \otimes \ghftp X \ghftp) \, 
\gamma_t^{[1](2)} \big] \:} 
\nonumber \\[1ex]
& \ = \
\Tr\big[ (\ghftp X \ghftp \otimes \ghf_t X \ghf_t) \, 
\gamma_t^{[1](2)} \big] 
\ = \
\Tr\big[ \GG(\ghftp X \ghftp \otimes \ghf_tX \ghf_t) 
\, \rhoone_t] \,.
\end{align}
Using Lemma~\ref{lem:dGamma2_w} yields in turn 
\begin{align} \label{eq:gpXgp-gXg-leq-}
\Tr\big[ \GG(\ghftp X \ghftp \otimes \ghf_tX \ghf_t) 
\, \rhoone_t ]
& \ \leq \ 
\Tr\big[ \ghf_t^{2} X] \, 
\Tr\big[ \G(\ghftp X \ghftp) \, \rhoone_t \big]
\nonumber \\[1ex]
& \ = \ 
\Tr[\ghf_t X] \, \Tr\big[X \ghftp \gamma^{[1]}_t \ghftp \big] \,.
\end{align}
For the second term on the right-hand side
of \eqref{eq:a(X)_first_estimate}, we observe that 
\begin{align*}
\Tr[ \G(\ghftp X \ghf_t^{2} X \ghftp) \, \rhoone_t ] 
\ = \ 
\Tr[X \ghf_t^{2} X \ghftp \gammaone_t \ghftp ]
\ \leq \
\Tr[\ghf_t X] \, \Tr[ X \ghftp \gammaone_t \ghftp ] \,,
\end{align*}
and for the third term on the right-hand side of
\eqref{eq:a(X)_first_estimate},
\begin{align*}
\Tr[\ghf_t X] \, 
\Tr\big[ \big( \Tr[\ghf_t X] - \G(\ghf_t X \ghf_t) \big) \, \rhomone_t \big] 
\ = \ 
\Tr[\ghf_t X] \, \Tr[(\Tr[ \rhomone_t]\ghf_t - \ghf_t \gammamone_t \ghf_t) \, X ] \,,
\end{align*}
which yields \eqref{eq:a(X)_leq_blabla}.
\end{proof}

We now give a bound on the integral $\int \mathcal{A}_t(X_{\omega}) d\mu(\omega)$ using the estimate from Proposition~\ref{pro:estimate_a(X)} on $\mathcal{A}_t(X)$. To get good estimates we take $g$ to be $g_\theta$ as in Eq.~\eqref{eq:g_theta}. We use the notation
\begin{align} \label{eq:f_HF}
f_{HF}(x) & \ := \ \ghf_t(x;x) \ \geq \ 0 \, ,
\end{align}
where $\ghfnot_t=\sum_{i=1}^N |\varphi_{i,t}\rangle\langle \varphi_{i,t}|$, with $\langle\varphi_{i,t}|\varphi_{j,t}\rangle=\delta_{ij}$ and $\ghf_t(x;y):=\sum_{i=1}^N \varphi_{i,t}(x) \overline{\varphi_{i,t}(y)}$, which allows us to rewrite the traces as integrals. For example,
\[
\Tr[\ghf_t X_{r,z}] \ = \ \int_{|x-z|\leq r} f_{HF}(x) \, d^3x \,.
\]
Observe that $\int f_{HF}=N$ and that the quantity $\int f_{HF}^{5/3}$ appearing in Proposition~\ref{pro:estimate-I} is controlled by the Lieb-Thirring inequality, as is discussed in Sect.~\ref{sec:proof_main_thm}.

Before we give the bound for $\mathcal{A}_t$, let us prove an auxiliary lemma. Let $A^c$ denote the complement of a set $A$ and recall that $B(0,R)$ denotes the ball of radius $R$ centered at $0$ in $\mathbb{R}^3$.

\begin{lem}\label{lem:Young-Holder-1}
For $\frac{1}{p_1}+\frac{1}{p_2}+\frac{1}{s}=2$, with $1\leq p_j,s\leq\infty$, measurable functions $\chi,f_1,f_2:\mathbb{R}^{3}\to\mathbb{R}$ and $v(x)=|x|^{-1}$,
\[
\int (\chi \, v)(x-y) \, f_1(x) \, f_2(y) \, d^3x \, d^3y
\ \leq \ 
\|f_1\|_{p_1} \, \|f_2\|_{p_2} \, \|\chi \, v\|_{s} \, .
\] 
Additionally, for $s<3$,
\[ 
\|1_{B(0,R)} \, v\|_{s}
\ = \ 
\Big(\frac{4\pi}{3-s}\Big)^{1/s} R^{3/s-1} \, , 
\] 
and, for $s>3$,
\[
\|1_{B(0,R)^c} \, v\|_{s}
\ = \ 
\Big(\frac{4\pi}{s-3}\Big)^{1/s} \, R^{3/s-1} \, , 
\] 
with the convention that $\Big(\frac{4\pi}{\infty-3}\Big)^{1/\infty} := 1$.
\end{lem}
\begin{proof}
The first relation follows directly from applying H\"older's and Young's inequalities. The second and third relations follow directly from integration.
\end{proof}

With the ingredients above we can give a bound on $\mathcal{A}_t$.

\begin{prop}\label{pro:estimate_A} \label{pro:estimate-I}
The estimate
\begin{equation}\label{eq:A_leq_blabla}
\mathcal{A}_t \leq 5^{-5/6} \, 72 \, \pi^{1/3} N^{1/6} \big\|f_{HF}\big\|_{5/3}^{5/6} \, S_{g_{\theta}}
\end{equation}
holds.
\end{prop}

\begin{proof}
By Proposition~\ref{pro:estimate_a(X)},
\begin{equation}\label{eq:pre_A_estimate}
\mathcal{A}_t \leq 2 \int \Tr[\ghf_t X_{\om}] \, \Tr[ \ghftp \gammaone_t \ghftp X_{\om} ] \, d\mu(\om) + \int \Tr[\ghf_t X_{\om}] \, \Tr[ \ghf_t {\gammamone_t}^{\perp} \ghf_t X_{\om} ] \, d\mu(\om) \, .
\end{equation}

We now explicitly use the Fefferman-de la Llave decomposition of the Coulomb potential. Then, we find that for any non-negative trace-class operator $h$, using $\int h(y;y)dy = \Tr[h]$, H\"older's inequality and Lemma~\ref{lem:Young-Holder-1} in the end (where we distinguish between the short-range and the long-range part of the potential) gives us that
\begin{align*}
\int \Tr[\ghf_t X_{\om}] \, \Tr[ h X_{\om} ] \, d\mu(\om) & = \frac{1}{\pi} \int \Tr[\ghf_t X_{r,z}] \, 
\Tr[h X_{r,z}] \, \frac{dr}{r^{5}} \, d^3z
\\ & \ = \ 
\frac{1}{\pi} \int 
\Big( \int_{|x-z| \leq r} f_{HF}(x) \, d^3x \Big) \,
\Big(\int_{|y-z| \leq r} h(y;y) \, d^3y \Big) \, \frac{dr}{r^{5}} \, d^3z
\\ & \ = \ 
\int \frac{1}{|x-y|} f_{HF}(x) \, h(y;y) \: d^3x \,d^3y 
\\ & \ \leq \ 
\Big( \|1_{B(0,R)} \, v\|_{5/2} \, \|f_{HF}\|_{5/3}
\: + \: \|1_{B(0,R)^c} \, v\|_{\infty} \, \|f_{HF}\|_{1} \Big)
\, 
\Tr[h]
\\[1ex] & \ \leq \ 
\big( (8\pi)^{2/5} R^{1/5} \|f_{HF}\|_{5/3} 
\: + \: R^{-1} \|f_{HF}\|_{1} \big) \, \Tr[h]
 \, .
\end{align*}
Optimizing with respect to $R>0$ yields
\[
R \ = \ (8\pi)^{-1/3} 5^{5/6} \|f_{HF}\|_{1}^{5/6} \, \|f_{HF}\|_{5/3}^{-5/6} \, ,
\]
so that (recall $\|f_{HF}\|_{1} = N$)
\begin{align*}
\int \Tr[\ghf_t X_{\om}] \, \Tr[ h X_{\om} ] \, d\mu(\om) &\leq 5^{-5/6} 6 (8\pi)^{1/3} \Big( \int f_{HF}^{5/3} \Big)^{1/2} \, \Big(\int f_{HF} \Big)^{1/6} \, \Tr[h] \\
& = 5^{-5/6} 12 \pi^{1/3} \, \|f_{HF}\|_{5/3}^{5/6} \, N^{1/6} \, \Tr[h] \, .
\end{align*}
We now apply this inequality to \eqref{eq:pre_A_estimate}, i.e., with $h=\ghftp \gammaone_t \ghftp$
and with $h=\ghf_t {\gammamone_t}^{\perp} \ghf_t $.
It follows from Lemma~\ref{lem:estimate_by_S_theta} that
\begin{align}\label{eq:int_f1_leq_S}
\Tr[\ghftp \gammaone_t \ghftp] & =\Tr[\ghftp \gammaone_t] = \Tr[\G(\ghftp) \rhoone_t]  \leq 2 S_{g_\theta} \,,
 \\  \label{eq:int_fprime_T_leq_S}
\Tr[\ghf_t {\gammamone_t}^\perp \ghf_t] & =\Tr[\ghft]\Tr[\rho_t^{[-1]}] - \Tr\big[ \G (\ghft) \rho_t^{[-1]}\big] = \Tr\big[ \G (\ghftp) \rho_t^{[-1]}\big]  \leq 2 S_{g_\theta} \, .
\end{align}  
This proves \eqref{eq:A_leq_blabla}.
\end{proof}

\subsection{Bound for $\mathcal{B}_t$}

We estimate $\mathcal{B}_t$ in the same fashion as in \cite[Lemma~7.3]{Petrat:2014fk}. Note, that for this term and the $\mathcal{C}_t$ term it is not necessary to use the Fefferman-de la Llave decomposition.

\begin{prop}\label{pro:estimate_B}
The estimate
\begin{equation}\label{eq:B_leq_blabla}
\mathcal{B}_t \leq 2^{1/3}  \pi^{2/3} \big\|f_{HF}\big\|_{5/3}^{5/6} \, N^{1/6} \big( 6S_{g_{\theta}} + N^{1-\theta} \big)
\end{equation}
holds.
\end{prop}

\begin{proof}
We estimate the $\mathcal{B}_t$ term by using the Cauchy-Schwarz inequality to arrive at a three-particle term.

Recall that $\rho_t^{[-2,2]}$ is a rank one operator, i.e., $\rho_t^{[-2,2]}=\big|\Psi^{[-2]}\big\rangle\big\langle \Psi^{[2]}\big|$. For a linear operator $A$ on $\fh$, $j\in\{-2,2\}$ and almost every~$x\in\mathbb{R}^3$, we define the vectors $\tilde{\Psi}^{[j]}_{A,x} \in \fhk{N-1}$ by
\[
\tilde{\Psi}^{[j]}_{A,x}(x_1,\dots,x_{N-1}):=(\Id^{\otimes N-1}\otimes A)\Psi^{[j]}(x_1,\dots,x_{N-1},x)\,.
\]
Inserting the form of $v^{(2)}$ as a direct integral into the expression of $\mathcal{B}_t$ in Eq.~\eqref{eq:b(X)} yields 
\begin{align}
\mathcal{B}_t & =  \Im	\Big\langle\Psi^{[-2]},\GG \big((\ghftp\otimes \Id)\int_{\mathbb{R}^3}^{\oplus}\ghftp(\tau_{x_{1}}v)\ghft\,dx_{1}(\ghft\otimes \Id)\big)\Psi^{[2]}\Big\rangle
 \nonumber \\
	&=N \Im\int_{\mathbb{R}^3}\Big\langle\tilde{\Psi}_{\ghftp,x}^{[-2]},\G \big(\ghftp(\tau_{x}v)\ghft\big) \, \tilde{\Psi}_{\ghft,x}^{[2]}\Big\rangle dx\,.
	 \nonumber 
\end{align}
Taking the modulus of both sides and using the Cauchy-Schwarz inequality we obtain
\begin{align}
\mathcal{B}_t & \leq \Big(N\int_{\mathbb{R}^3}\big\|\tilde{\Psi}_{\ghftp,x}^{[-2]}\big\|^{2}dx\Big)^{1/2}\Big(N\int_{\mathbb{R}^3} \big\|\G(\ghftp(\tau_{x}v)\ghft) \, \tilde{\Psi}_{\ghft,x}^{[2]}\big\|^{2}dx\Big)^{1/2} \nonumber \\
&= \big\langle\Psi^{[-2]},\G (\ghftp)\Psi^{[-2]} \big\rangle^{1/2}\Big(N\int_{\mathbb{R}^3} \big\langle\tilde{\Psi}_{\ghft,x}^{[2]},\G(\ghft(\tau_{x}v)\ghftp) \, \G(\ghftp(\tau_{x}v)\ghft) \, \tilde{\Psi}_{\ghft,x}^{[2]}\big\rangle dx\Big)^{1/2}\,. \nonumber 
\end{align}
From Lemma \ref{lem:estimate_by_S_theta} we deduce that $\langle\Psi^{[-2]},\G(\ghftp)\,\Psi^{[-2]}\rangle\leq 6 S_{g_\theta}$.
We estimate the remaining integral using Eq.~\eqref{eq:dGdG_dG2+dG}, the Cauchy-Schwarz inequality, Lemma~\ref{lem:dGamma2_w} and Lemma~\ref{lem:estimate_by_S_theta}:
\begin{align}
N\int_{\mathbb{R}^3} & \Big\langle\tilde{\Psi}_{\ghft ,x}^{[2]},\Big[\GG\big((\ghft(\tau_{x}v)\ghftp)\otimes(\ghftp(\tau_{x}v)\ghft)\big)+ \G\big(\ghft(\tau_{x}v)\ghftp(\tau_{x}v)\ghft\big)\Big] \tilde{\Psi}_{\ghft,x}^{[2]}\Big\rangle dx \nonumber \\
&\leq N\int_{\mathbb{R}^3} \Big\langle\tilde{\Psi}_{\ghft,x}^{[2]},\Big[\GG\big((\ghft(\tau_{x}v)^{2}\ghft)\otimes \ghftp\big)+\G\big(\ghft(\tau_{x}v)^{2}\ghft\big)\Big]\tilde{\Psi}_{\ghft,x}^{[2]}\Big\rangle dx \nonumber \\
&=\Big\langle\Psi^{[2]},\Big[\GG\big(\ghft^{\otimes2}(v^{2})^{(2)}\ghft^{\otimes2}\big) \G(\ghftp) + \GG\big(\ghft^{\otimes2}(v^{2})^{(2)}\ghft^{\otimes2}\big)\Big]\Psi^{[2]}\Big\rangle \nonumber \\
&\leq \Big(\langle\Psi^{[2]},\G(\ghftp)\Psi^{[2]}\rangle+\langle\Psi^{[2]},\Psi^{[2]}\rangle \Big) \int_{\mathbb{R}^3}(f_{HF}*v^{2})f_{HF}  \nonumber \\
&\leq (6S_{g_\theta}+2N^{1-\theta}) \int_{\mathbb{R}^3}(f_{HF}*v^{2})f_{HF}\,. \nonumber 
\end{align}

By the Hardy-Littlewood-Sobolev inequality (see, e.g., \cite[Theorem~4.3]{liebloss:2001}) we find
\begin{equation}\label{eq:HLS}
\big\|(v^2* f)f\big\|_1 \leq 4^{1/3} \pi^{4/3} \|f\|_{3/2}^2.
\end{equation}
We then apply H\"older's inequality with $1=\frac{3}{4}+\frac{1}{4}$ to obtain
\begin{equation}\label{eq:HLS_Hoelder}
\|f\|_{3/2}^2 = \|f^{5/4} f^{1/4}\|_1^{4/3} \leq \|f^{5/4}\|_{4/3}^{4/3} \, \|f^{1/4}\|_4^{4/3} = \|f\|_{5/3}^{5/3} \, \| f \|_1^{1/3} \, .
\end{equation}
Applying this to $f_{HF}$ and $ab\leq (a^2+b^2)/2$ to the bound we obtained on $\mathcal{B}_t$ yields the result.
\end{proof}

\subsection{Bound for $\mathcal{C}_t$}\label{sec:estimate_C}

Our estimate for $\mathcal{C}_t$ is analogous to \cite[Lemma~7.3]{Petrat:2014fk}. Note that for this estimate our choice of the function $g_{\theta}$ is crucial, while in the bounds for $\mathcal{A}_t$ and $\mathcal{B}_t$ we could have used the identity function to obtain the desired estimate. By using $g_{\theta}$ with appropriate $\theta<1$ we obtain the desired $N$-dependence in the estimate for $\mathcal{C}_t$.

\begin{prop}
\label{pro:estimate_C}  The estimate
\begin{align}
\label{eq:c(X)_leq_blabla}
\mathcal{C}_t \leq 4 \sqrt{2} \, \big\|f_{HF}* v^2\big\|_{\infty}^{1/2} N^{\theta/2} S_{g_{\theta}}
\end{align}
holds.
\end{prop}

\begin{proof}
Using the Cauchy-Schwarz inequality, Lemma~\ref{lem:dGamma2_w} and Lemma~\ref{lem:estimate_by_S_theta} we find
\begin{align*}
\mathcal{C}_t & = 2\Im \Tr\big[\GG\big((\ghftp \otimes\ghftp )v^{(2)}(\ghft\otimes\ghftp )\big)\, \rho_t^{[-1,1]}\big] \nonumber \\
& \leq 2  \Big( \, \Tr\big[\GG\big((\ghft\otimes\ghftp )\big(v^{(2)}\big)^2(\ghft\otimes\ghftp )\big)\, \rho_t^{[-1]}\big]
\Tr\big[\GG\big(\ghftp \otimes\ghftp \big)\, \rho_t^{[1]}\big] \Big)^{1/2} \nonumber \\
& \leq 2 \Big( \big\|f_{HF}* v^2\big\|_{\infty} \, \Tr\big[\G\big(\ghftp \big)\, \rho_t^{[-1]}\big] 4 N^{\theta} S_{g_{\theta}} \Big)^{1/2} \nonumber \\
& \leq 4 \sqrt{2} \, \big\|f_{HF}* v^2\big\|_{\infty}^{1/2} N^{\theta/2} S_{g_{\theta}} \, ,
\end{align*}
which is the result.
\end{proof}

\subsection{Kinetic Energy Estimates and Proof of Theorem~\ref{thm:estimate-S}}\label{sec:proof_main_thm}

In order to estimate $\|f_{HF}\|_{5/3}^{5/3}$ and $\|f_{HF}* v^2\|_{\infty}$ in terms of the kinetic energy we use the following inequalities.
\begin{prop}[Lieb-Thirring Inequality, see \cite{PhysRevLett.35.687} or {\cite[p.73]{MR2583992}}
]\label{pro:Lieb-Thirring}
Let $\gamma \in \mathcal{L}^{1}(\fh)$ be a one-particle density matrix
of finite kinetic energy, i.e., $0 \leq \gamma \leq \Id$ and
$\Tr[-\Delta\gamma]<\infty$. Then
\begin{equation}\label{eq:Lieb-Thirring}
C_{LT} \int f^{5/3}(x) \, d^3x 
\ \leq \
\Tr[-\Delta\gamma] \, ,
\end{equation}
with $C_{LT}=\frac{9}{5}(2\pi)^{2/3}$ and where $f(x):=\gamma(x;x)$ is the corresponding one-particle density.
\end{prop}

\begin{prop}[Hardy's Inequality, see \cite{hardy:1920} or, e.g., \cite{MR1717839}]\label{pro:hardy}
Let $\gamma \in \mathcal{L}^{1}(\fh)$ be a one-particle density matrix
of finite kinetic energy, i.e., $0 \leq \gamma \leq \Id$ and
$\Tr[-\Delta\gamma]<\infty$. Then 
\begin{equation}\label{eq:hardy}
\int \frac{f(x)}{|x|^2} \, d^3x \leq 4 \, \Tr[-\Delta\gamma] \, ,
\end{equation}
where $f(x):=\gamma(x;x)$ is the corresponding one-particle density.
\end{prop}

We now combine the results of Sects.~\ref{sec:time-derivative-S} to \ref{sec:estimate_C} to prove Theorem~\ref{thm:estimate-S}.

\begin{proof}[Proof of Theorem~\ref{thm:estimate-S}]
We choose $\theta = \frac{1}{3}$ so that our bound for $\mathcal{C}_t$ is good enough. Collecting the estimates for the $\mathcal{A}_t$, $\mathcal{B}_t$ and $\mathcal{C}_t$ terms from Propositions~\ref{pro:estimate-I}, \ref{pro:estimate_B} and \ref{pro:estimate_C} and using the kinetic energy inequalities from Propositions~\ref{pro:Lieb-Thirring} and \ref{pro:hardy}, we can continue the estimate for the time derivative of $S_{g_{1/3},t}$ from Proposition~\ref{pro:equation-dS_g/dt} and find (recall that $K := \sup_{t \geq 0} \Tr[-\Delta \ghf_t]$)
\begin{align}
\frac{dS_{g_{1/3}}(t)}{dt} & = \lambda \big( \mathcal{A}_t + \mathcal{B}_t + \mathcal{C}_t \big) \nonumber \\
& \leq \lambda \, 5^{-5/6} \, 72 \, \pi^{1/3} \, N^{1/6} \, \left( \frac{5}{9} (2\pi)^{-2/3} K \right)^{1/2} \, S_{g_{1/3}}(t) \nonumber \\
& \quad + \lambda \, 2^{1/3}  \pi^{2/3} \left( \frac{5}{9} (2\pi)^{-2/3} K \right)^{1/2} \, N^{1/6} \Big( 6S_{g_{1/3}}(t) + N^{2/3} \Big) \nonumber \\
& \quad + \lambda \, 4 \sqrt{2} \, \left( 4 K \right)^{1/2} N^{1/6} S_{g_{1/3}}(t) \nonumber \\
& \leq 30 \, \lambda \, \sqrt{K} \, N^{1/6} \Big( S_{g_{1/3}}(t) + N^{2/3} \Big) \, .
\end{align}
Integrating this inequality (Gr\"onwall lemma) yields Theorem~\ref{thm:estimate-S}.
\end{proof}

\appendix

\section{Some Results about the Theory of \\ the Time-Dependent Hartree-Fock
Equation}\label{sec:Theory_TDHF}

In this appendix we recall some known facts about the theory of the
TDHF equation. We begin by stating a theorem summarizing those
results proved in \cite{MR0456066} which we use.
\begin{thm}
Let $E$ a separable Hilbert space, $A:E\supseteq\mathcal{D}(A)\to E$
self-adjoint such that $\exists\mu\in\mathbb{R}$, $A\geq\mu\, \Id$.  Let
$M:=(A-\mu+\Id)^{1/2}$ and\[ H_{k,p}^{A}(E):=\big\{M^{-k}TM^{-k} \;
\big| \; T=T^{*}\,,\quad T\in\mathcal{L}^{p}(E)\big\}\,,\] equipped
with the norm $\|T\|_{k,p,A}=\|M^{k}TM^{k}\|_{p}$ where
$\|X\|_{p}=\Tr[|X|^{p}]^{1/p}$ for $1\leq p<\infty$ or
$\|X\|_{\mathcal{B}(E)}$ for $p=\infty$ (we write
$\mathcal{L}^{\infty}(E)$ for $\mathcal{B}(E)$). We adopt the special
notations $H(E):=H_{0,\infty}^{A}(E)$ for the space of bounded
self-adjoint operators on $E$ and $H_{1}^{A}(E):=H_{1,1}^{A}(E)$ for a
weighted space of trace-class operators on $E$.

Let $\mathcal{W}\in\mathcal{B}(H_{1}^{A}(E);H(E))$ be such that
\begin{enumerate}
\item $\big(\mathcal{W}(T)M^{-1}\big)(E)\subseteq\mathcal{D}(M)$,
\item $\big(T\mapsto M\mathcal{W}(T)M^{-1}\big)\in\mathcal{B}(H_{1}^{A}(E);H(E))$,
\item $\forall T,S\in H_{1}^{A}(E): \ \ $
  $\Tr[\mathcal{W}(T)S]=\Tr[\mathcal{W}(S)T]$.
\end{enumerate}

Then
\begin{itemize}
\item For any $T_{0}\in H_{1}^{A}(E)$ there exists $t_{0}>0$ and $T\in C([0,t_{0});H_{1}^{A}(E))$
such that, $\forall t\in[0,t_{0})$,\[
T(t)=e^{-itA}\,T_{0}\,e^{itA} -i \int_{0}^{t}e^{-i(t-s)A} \, \big[\mathcal{W}(T(s)),T(s)\big] \, e^{i(t-s)A}\, ds\,.\]
Such a function $T$ is called a local \textit{mild solution} of the
TDHF equation and, provided its interval of definition is maximal,
it is unique.
\item If, moreover, $T_{0}\in H_{2,1}^{A}(E)$ then $T\in C^{1}([0,t_{0});H_{1}^{A}(E))$
and\[
\begin{cases}
i\frac{dT}{dt}(t) & =\ \big[A,T(t)\big] \ +\ \big[\mathcal{W}(T(t))\,,\, T(t)\big]\\
T(0) & =\ T_{0}\end{cases}\,.\]
Such a function $T$ is called a \textit{classical solution} of the
TDHF equation.
\item Any mild solution to the TDHF equation satisfies\[
\forall t\in[0,t_{0})\,,\quad \Tr\big[MT(t)M\big]+\frac{1}{2}\Tr\big[T(t)\, \mathcal{W}(T(t))\big]=\Tr\big[MT_{0}M\big]+\frac{1}{2}\Tr\big[T_{0}\, \mathcal{W}(T_{0})\big]\,.\]

\item If $\exists k_{1}\in\mathbb{R}$ such that\footnote{There was a typographical error in Assumption $iv)$ in \cite{MR0456066}, namely, 
$\mathcal{W}(T)T\geq k_{1}$ shall be read $\mathcal{W}(T)\geq k_{1}$.
} \[
\big(T\in H_{1}^{A}(E)\,,\,0\leq T\leq \Id\big)\quad\Rightarrow\quad\big(\mathcal{W}(T)\geq k_{1}\big)\,,\]
and $T_{0}\in H_{1}^{A}(E)$, $0\leq T_{0}\leq \Id$, then $T$ can be
extended to the entire positive real axis. Moreover if $T_{0}\in H_{A}^{2,1}(E)$,
then $T$ is the unique global classical solution.
\end{itemize}
\end{thm}
\begin{rem}
In \cite{MR0456066} the space $H_{2,1}^{A}(E)$ is not used. They
use a space larger than $H_{2,1}^{A}(E)$ which is more natural, but
less explicit. As it is enough for us to use classical solutions for
initial data in $H_{2,1}^{A}(E)$ and then use a density result, we
restrict ourselves to this framework.
\end{rem}
We now quote a result which, although not explicitly
stated in \cite{MR0456066}, is a direct consequence of
\cite{MR0456066} along with \cite{MR0152908}.
\begin{prop}
The map
\begin{align*}
H_{1}^{A}(E)\times[0,\infty) & \to H_{1}^{A}(E)\\
(T_{0},t) & \mapsto T(t) \, ,
\end{align*}
where $T(t)$ is the (mild) solution to the TDHF equation with initial
data $T_{0}$, is jointly continuous in $T_{0}$ and $t$.
\end{prop}
Indeed the proof of existence and uniqueness in \cite{MR0456066}
is based on the results in \cite{MR0152908} which also ensure the
continuity with respect to the initial data (see \cite[Corollary~1.5, p.~350]{MR0152908}).

It was shown in \cite{MR0456066} that those results apply to the
case $E=\mathfrak{h}=L^{2}(\mathbb{R}^{3})$, $A=-\Delta$, 
\[
\mathcal{W}(\gamma) \ = \ \Tr_{2}\big[ v^{(2)} (\Id-\Ex) (\Id \otimes \gamma) \big] \, ,
\]
and $v^{(2)}=|x-y|^{-1}$. The proof then extends to the case
$A=h^{(1)}$ with $h^{(1)}=-C\Delta+w(x)$, where the external potential
$w$ is an infinitesimal perturbation of the Laplacian.

\section{Some Estimates of the Direct Term and the Kinetic Energy}\label{sec:misc_appendix}

\paragraph{The dynamics is the free dynamics to leading order in the $\lambda K^{1/2}N^{1/6}t \sim 1$ regime.}
We now substantiate by heuristic argument that, in the particular case of the Coulomb interaction potential, if $\lambda K^{1/2}N^{1/6}t$ is assumed to be of order one, which is the regime where our estimates are relevant, then the evolution is the free evolution to leading order. 
Note that the exchange term is expected to be subleading with respect to the direct term; we thus neglect the exchange term in the following computation.

We now estimate the effect of the direct term on the time derivative of the average momentum per particle. We denote the Hartree-Fock density at time $t$ by $f_{HF,t}=\sum_{j=1}^N|\varphi_{t,j}|^2$; thus the direct term is the convolution $\lambda v*f_{HF,t}$. For the absolute value of the time derivative of the expectation value of the momentum per particle we find, using 
$|\nabla v| = 3 v^2$ and \eqref{eq:HLS} with \eqref{eq:HLS_Hoelder},
\begin{align}
\left| N^{-1} \partial_t \Tr \big(p_t (-i\nabla)\big) \right| &= N^{-1} \left| \Tr \big(p_t [h_{HF},\nabla]\big) \right| \nonumber \\
&\leq \lambda N^{-1} \left| \Tr \big(p_t (\nabla v * f_{HF,t})\big) \right| \nonumber \\
&\leq 3 \lambda N^{-1} \big\| (v^2 * f_{HF,t}) f_{HF,t} \big\|_1 \nonumber \\
&\leq C \lambda N^{-2/3} K.
\end{align}
After a time $t$ the effect of the direct term on the momentum is thus expected to be of order $\lambda N^{-2/3} K t$. Since $\lambda K^{1/2}N^{1/6}t$ is assumed to be of order one, the average change in momentum is of order $K^{1/2}N^{-5/6}$. Since this is much smaller than the average momentum of a particle $(K/N)^{1/2}$, we conclude that the dynamics is, to leading order, free.

\paragraph{The estimate \eqref{eq:main_estimate} allows to distinguish the free dynamics and the Hartree-Fock dynamics to the next order.}

Again using heuristic arguments we substantiate that, for large enough kinetic energy $K\gg N^{4/3}$, estimate~\eqref{eq:main_estimate} allows to distinguish the effect of the direct term on the evolution, i.e., our main result shows that the Hartree-Fock equation gives a better approximation to the Schr\"odinger equation than the free equation. This is because our convergence rate is $N^{-1/6}$ (let us assume $\gamma_0=p_0$), i.e., for $\lambda K^{1/2}N^{1/6}t$ of order $1$, the error between Schr\"odinger and Hartree-Fock dynamics is for any bounded observable of order $N^{-1/6}$. For $K\gg N^{4/3}$ this rate is much smaller than the average change in momentum estimated above, i.e., $N^{-1/6} \ll K^{1/2}N^{-5/6}$.

\paragraph{Estimate on the Kinetic Energy.} We prove the estimate \eqref{eq:bound_on_e_kin}. 
Recall that $f_{HF,t}=\sum_{j=1}^N|\varphi_{t,j}|^2$ and that the direct term is the convolution $\lambda v*f_{HF,t}$. 
First, we use the conservation of the total Hartree-Fock energy and the fact that the exchange energy is bounded by the direct energy (which follows directly from applying the Cauchy-Schwarz inequality). Then, by H\"older's inequality, we find
\begin{align*}
\Tr[-\Delta \ghft] &\leq \mathcal{E}(\Phi_{HF,0}) + \frac{\lambda}{2} \int (v* f_{HF,t})(x) f_{HF,t}(x) d^3x \\
&\leq \mathcal{E}(\Phi_{HF,0}) + \frac{\lambda}{2} \|v* f_{HF,t}\|_{\infty}\, N \,.
\end{align*}
By splitting the area of integration in space  in the convolution  product $v*f_{HF,t}$ into a short-range and a long-range part, using H\"older's inequality, Lemma~\ref{lem:Young-Holder-1} and the Lieb-Thirring inequality, we find
\begin{align}\label{eq:v_star_rho_infinity_proof}
\big\|v*f_{HF,t}\big\|_{\infty} &\leq \| f_{HF,t} \|_{5/3} \|1_{B(0,R)} \, v\|_{5/2} + \| f_{HF,t} \|_1 \| 1_{B(0,R)^c} v \|_{\infty} \nonumber \\
&\leq C K^{3/5}R^{1/5} + N R^{-1} \, .
\end{align}
Optimizing with respect to $R$ gives the bound 
\begin{equation}\label{eq:v_star_rho_infinity}
\big\|v*f_{HF,t}\big\|_{\infty} \leq C K^{1/2}N^{1/6} \,.
\end{equation}
Then, by Eq.~\eqref{eq:v_star_rho_infinity} and $ab\leq \frac{a^2}{2} + \frac{b^2}{2}$, we find
\begin{align*}
\Tr[-\Delta \ghft] 
&\leq \mathcal{E}(\Phi_{HF,0}) + C \lambda \left(\Tr[-\Delta \ghft]\right)^{1/2} N^{1/6} N \\
&\leq \mathcal{E}(\Phi_{HF,0}) + \frac{1}{2} \Tr[-\Delta \ghft] + C \lambda^2 N^{7/3} \,,
\end{align*}
which proves \eqref{eq:bound_on_e_kin}.

\section*{Acknowledgments}

The authors are indebted to N.~Benedikter, M.~Porta and B.~Schlein for helpful discussions and sharing their results prior to publication. S.~B.'s research is supported by the Basque Government through the BERC 2014-2017 program and by the Spanish Ministry of Economy and Competitiveness MINECO: BCAM Severo Ochoa accreditation SEV-2013-0323. S.~P.'s research has received funding from the People Programme (Marie Curie Actions) of the European Union's Seventh Framework Programme (FP7/2007-2013) under REA grant agreement n\textdegree~291734. T.~T.\ is supported by the DFG Graduiertenkolleg 1838, and for part of this work was supported by DFG Grant No.\ Ba-1477/5-1 and also by the European Research Council under the European Community's Seventh Framework Program (FP7/2007-2013)/ERC grant agreement 20285.

\bibliographystyle{plain}
\bibliography{Biblio}

\end{document}